\documentclass[11pt]{article}
\usepackage{geometry}                
\usepackage{graphicx}
\usepackage{amssymb}
\usepackage{epstopdf}
\usepackage{fleqn}    
\usepackage{booktabs} 
\usepackage{multirow} 
\usepackage{subfig}

\usepackage{color}
\usepackage{ifthen}
\usepackage{eucal}

\usepackage{graphicx}

\usepackage[fleqn]{amsmath}
\usepackage{amssymb}
\usepackage{amsbsy}

\usepackage{algorithm}
\usepackage{algorithmicx}
\usepackage{algpseudocode} 
\usepackage{natbib}
\usepackage{amsthm}
\newtheorem{theorem}{Theorem}
\newtheorem{definition}{Definition}

\newtheorem{lemma}[theorem]{Lemma}

\usepackage{url}

\usepackage{bm}

\usepackage{pdfpages}

\usepackage{titletoc}

    \newcommand{\Cc}{\mathcal{C}}

    \newcommand{\Fc}{\mathcal{F}}

    \newcommand{\Ic}{\mathcal{I}}

    \newcommand{\Lc}{\mathcal{L}}

    \newcommand{\Pc}{\mathcal{P}}

    \newcommand{\Uc}{\mathcal{U}}

\newcommand{\R}{\mathbb{R}}

\DeclareMathOperator*{\argmin}{argmin}
\DeclareMathOperator*{\argmax}{argmax}

\newcommand{\Algo}[1]{Algorithm~\protect\ref{#1}}

\newcommand{\Fig}[1]{Figure~\protect\ref{#1}}

\newcommand{\tab}[1]{Table~\protect\ref{#1}}

\newcommand{\per}{\mathbf{PER}}
\newcommand{\rej}{\mathbf{REJ}}
\newcommand{\astar}{\mathbf{A^{\ast}}}
\newcommand{\osstar}{\mathbf{OS^{\ast}}}
\newcommand{\Uniform}{\mathrm{Uniform}}
\newcommand{\Poissonmad}{\mathrm{Poisson}}
\newcommand{\Normal}{\mathrm{Normal}}
\newcommand{\Exponential}{\mathrm{Exp}}
\newcommand{\Bernoulli}{\mathrm{Bernoulli}}
\newcommand{\Gumbel}{\mathrm{Gumbel}}

\newcommand{\TruncGumbel}{\mathrm{TruncGumbel}}

\newcommand{\sort}{\mathrm{sort}}
\newcommand{\shear}{\mathrm{perturb}}
\newcommand{\keep}{\mathrm{accept}}
\newcommand{\thin}{\mathrm{thin}}
\newcommand{\expectmad}{\mathbb{E}}
\newcommand{\proba}{\mathbb{P}}

\title{A Poisson process model for Monte Carlo}
\author{Chris J. Maddison\\
Department of Computer Science\\
University of Toronto\\
cmaddis@cs.toronto.edu}

\begin{document}
\maketitle

\begin{abstract}
Simulating samples from arbitrary probability distributions is a major research program of statistical computing. Recent work has shown promise in an old idea, that sampling from a discrete distribution can be accomplished by perturbing and maximizing its mass function. Yet, it has not been clearly explained how this research project relates to more traditional ideas in the Monte Carlo literature. This chapter addresses that need by identifying a Poisson process model that unifies the perturbation and accept-reject views of Monte Carlo simulation. Many existing methods can be analyzed in this framework. The chapter reviews Poisson processes and defines a Poisson process model for Monte Carlo methods. This model is used to generalize the perturbation trick to infinite spaces by constructing Gumbel processes, random functions whose maxima are located at samples over infinite spaces. The model is also used to analyze A* sampling and OS*, methods from distinct Monte Carlo families.
\end{abstract}

\section{Introduction}
The simulation of random processes on computers is an important tool in scientific research and a subroutine of many statistical algorithms. One way to formalize this task is to return samples from some distribution given access to a density or mass function and to a pseudorandom number generator that returns independent uniform random numbers. ``Monte Carlo methods'', a phrase originally referring to the casinos of Monte Carlo, is a catchall for algorithms that solve this problem. Many Monte Carlo methods exist for specific distributions or classes of distributions \citep{walker1977alias, devroye}, but there are a few generic principles. One principle is to simulate a Markov chain whose stationary distribution is the distribution of interest.
Work on these Markov chain Monte Carlo methods has exploded over the past few decades, because of their efficiency at sampling from complex distributions in high dimensions. Their downside is that convergence can be slow and detecting convergence is hard. A second principle is propose samples from a tractable distribution and accept them according to a correction factor. These accept-reject Monte Carlo methods are the workhorses of modern statistical packages, but their use is restricted to simple distributions on low dimensional spaces. 

Recently, a research program has developed around another principle for sampling from discrete distributions, the so called ``Gumbel-Max trick''. The trick proceeds by simulating a random function $G : \{1, \ldots, m\} \rightarrow \R$ whose maximum is located at a sample. Sampling therefore reduces to finding the state that maximizes $G$. This trick has the same complexity as better known methods, but it has inspired research into approximate methods and extensions. Methods that abandon exactness for efficiency have considered introducing correlated $G$ with a variety of applications (\citeauthor{papandreou2011perturb}, \citeyear{papandreou2011perturb}; \citeauthor{tarlow2012randomized}, \citeyear{tarlow2012randomized}; \citeauthor{hazan2013perturb}, \citeyear{hazan2013perturb}). \cite{2015arXiv150609039C} consider bandit algorithms for optimizing $G$ over low dimensional spaces when function evaluation is expensive. \cite{maddison2014astarsamp} generalized $G$ with Gumbel processes, random functions over infinite spaces whose maxima occur at samples of arbitrary distributions, and introduced A* sampling, a branch and bound algorithm that executes a generalized Gumbel-Max trick. \cite{kim2016lprelaxsamp} introduced a related branch and bound algorithm tailored to discrete distributions and successfully sampled from a large fully connected attractive Ising model.
Taken together, this view of simulation as a maximization problem is a promising direction, because it connects Monte Carlo research with the literature on optimization.
Yet, its relationship to more established methods has not been clearly expressed.
This chapter addresses that need by identifying a model that jointly explains both the accept-reject principle and the Gumbel-Max trick. 

As a brief introduction, we cover a simple example of an accept-reject algorithm and the Gumbel-Max trick shown in \Fig{fig:intro}. Suppose we are given a positive function $f : \{1, \ldots, m\} \to \R^+$, which describes the unnormalized mass of a discrete random variable $I$,
\begin{align}
\label{eq:example}
\proba(I \in B) = \sum_{i \in B} \frac{f(i)}{\sum_{j=1}^m f(j)}, \quad B \subseteq \{1, \ldots, m\}.
\end{align}
The following algorithms return an integer with the same distribution as $I$. The accept-reject algorithm is,
\begin{enumerate}
\item Sample $J$ uniformly from $\{1, \ldots, m\}$, $U$ uniformly from $[0,\max_{i=1}^m f(i)]$,
\item If $U < f(J)$, return $J$, else go to 1.
\end{enumerate}
We can intuitively justify it by noticing that accepted pair $(J, U)$ falls uniformly under the graph of $f(i)$, \Fig{fig:intro}. The sample $J$, which is accepted or rejected, is often called a \emph{proposal}. The Gumbel-Max trick proceeds by optimizing a random function,
\begin{enumerate}
\item For $i \in \{1, \ldots m\}$ sample an independent Gumbel random variable $G(i)$.
\item Find and return $I^* = \argmax_{i=1}^m \log f(i) + G(i)$.
\end{enumerate}
Because the random values $\log f(i) + G(i)$ can be seen as a perturbed negative energy function, the function $G$ is often called a \emph{perturbation}. Uniform and Gumbel random variables are included among the standard distributions of statistical computing packages. So these algorithms, while inefficient, are simple to program.

\begin{figure}[t]
\begin{center}
\includegraphics{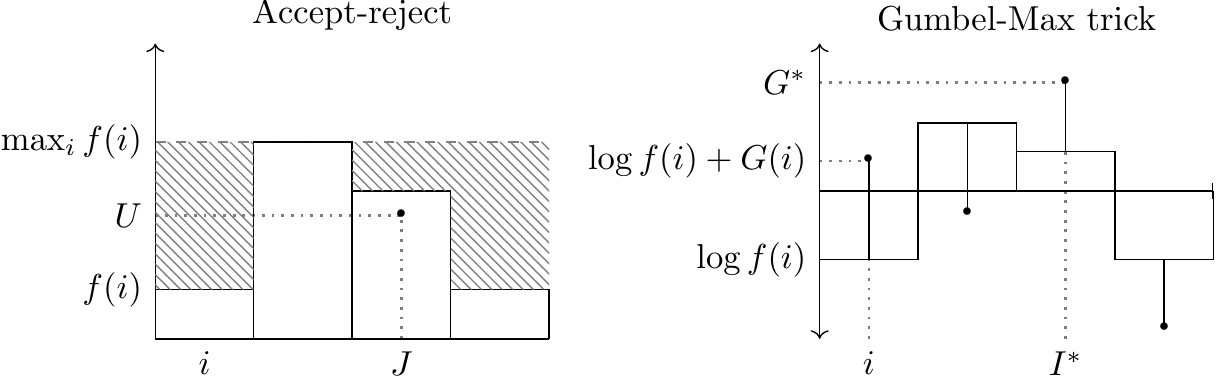}
\caption{Two simple Monte Carlo methods for a discrete distribution described by positive function $f$ via (\ref{eq:example}). The left hand plot shows the first accepted sample $J$ in an accept-reject scheme; note that $U < f(J)$. The right hand plot shows a sample $I^*$ in the Gumbel-Max trick; $I^*$ is the state that achieves the maximum $G^* = \max_i \log f(i) + G(i)$.}
\label{fig:intro}
\end{center}
\end{figure}

Considering their apparent differences and the fact that they have been studied in distinct literatures, it is surprising that both algorithms can be unified under the same theoretical framework.
The framework rests on the study of Poisson processes, a random object whose value is a countable set of points in space \citep{kingman1992poisson, daley2007introduction}. The central idea is to define a specific Poisson process, called an exponential race, which models a sequence of independent samples arriving from some distribution. Then we identify two operations, corresponding to accept-reject and the Gumbel-Max trick, which modify the arrival distribution of exponential races. In this view a Monte Carlo method is an algorithm that simulates the first arrival of an exponential race, and many existing algorithms fall into this framework. 

Section \ref{sec:pp} reviews Poisson processes and studies the effect of operations on their points. Section \ref{sec:er} introduces exponential races and studies the accept-reject and perturb operations. In Section \ref{sec:gup} we construct Gumbel processes from exponential races and study the generalized Gumbel-Max trick. In Section \ref{sec:alg} we analyze A* sampling and OS* \citep{dymetman2012osstar} and show how they use perturb and accept-reject operations, respectively, to simulate the first arrival of an exponential race. All of our Poisson process results are either known or elementary extensions; the correctness and behaviour of the Monte Carlo methods that we study have all been established elsewhere. Our contribution is in identifying a theory that unifies two distinct literatures and in providing a toolset for analyzing and developing Monte Carlo methods.

\section{Poisson processes}
\label{sec:pp}
\subsection{Definition and properties}
A Poisson process is a random countable subset $\Pi \subseteq \R^n$. Many natural processes result in a random placement of points: the stars in the night sky, cities on a map, or raisins in oatmeal cookies. A good generic mental model to have is the plane $\R^2$ and pinpricks of light for all points in $\Pi$. Unlike most natural processes, a Poisson process is distinguished by its complete randomness; the number of points in disjoint subsets are independent random variables, see \Fig{fig:pp}. In this section we review a general Poisson process theory culminating in two theorems, which describe how they behave under the generic operations of removing or relocating their points. In the next section we restrict our view to a specific Poisson process and two specific operations, which correspond to accept-reject and Gumbel-Max. Our study is situated in $\R^n$ for intuition, but these results generalize naturally; for more information, the ideas of this section are adapted from the general treatment in \cite{kingman1992poisson}. Readers familiar with that treatment can safely skip this section

\begin{figure}[t]
\begin{center}
\includegraphics{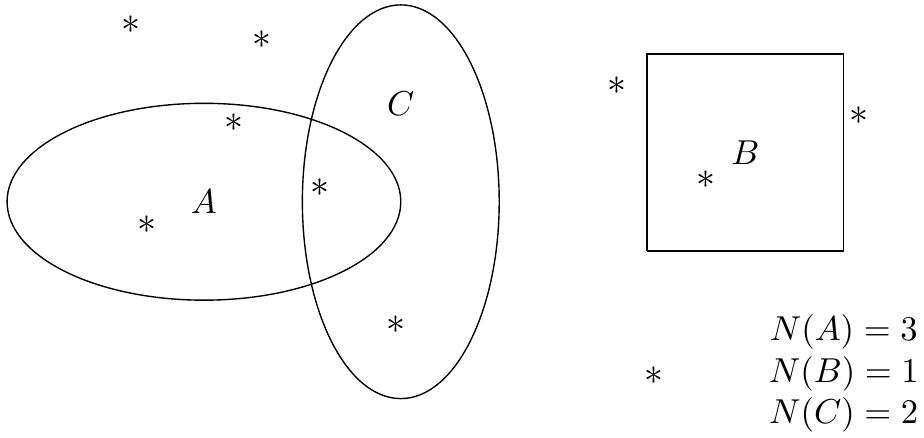}
\caption{The set of $\ast$ is a realization of a Poisson process in the plane. Counts in sets $A, B, C$ are marginally Poisson and are independent for disjoint sets.}
\label{fig:pp}
\end{center}
\end{figure}

To identify a realization of a random countable set $\Pi \subseteq \R^n$, we use counts of points in subsets $B \subset \R^n$,
\begin{align*}
N(B) = \# (\Pi \cap B).
\end{align*}
where $N(B) = \infty$ if $B$ is infinite, see \Fig{fig:pp} again. Counts are nonnegative and additive, so for any realization of $\Pi$ $N(B)$ satisfies
\begin{enumerate}
\item\label{itm:nonneg} (\emph{Nonnegative}) $N(B) \geq 0$,
\item\label{itm:countadd} (\emph{Countably additive}) For disjoint $B_i \subseteq \R^n$, $N(\cup_{i=1}^{\infty} B_i) = \sum_{i=1}^{\infty} N(B_i).$
\end{enumerate}
Set functions from subsets of $\R^n$ to the extended reals $\R \cup \{\infty, -\infty\}$ that are nonnegative and countably additive are called measures. Measure theory is a natural backdrop for the study of Poisson processes, so we briefly mention some basic concepts. In general measures $\mu$ assign real numbers to subsets with the same consistency that we intuitively expect from measuring lengths or volumes in space. If $\mu(\R^n) = 1$, then $\mu$ is a probability distribution. Because it is not possible to define a measure consistently for all possible subsets, the subsets $B \subseteq \R^n$ are restricted here and throughout the chapter to be from the Borel sets, a nice measurable family of subsets. The Borel sets contain almost any set of interest, so for our purposes it is practically unrestricted. Integration of some function $f : \R^n \to \R$ with respect to some measure $\mu$ naturally extends Riemann integration, which we can think about intuitively as the area under the graph of $f(x)$ weighted by the instantaneous measure $\mu(dx)$.
When a measure is equal to the integral of a nonnegative function $f: \R^n \to \R^{\geq 0}$ with respect to $\mu$, we say $f$ is the \emph{density} with respect to $\mu$. 

The Poisson process receives its name from the marginal distribution of counts $N(B)$. $N(B)$ is Poisson distributed on the nonnegative integers parameterized by a rate, which is also its expected value.
\begin{definition}[Poisson random variable]
$N$ is a Poisson distributed random variable on $k \in \{0, 1, \ldots\}$ with nonnegative rate $\lambda \in \R^{\geq 0}$ if
\begin{align*}
\proba(N = k) = \exp(-\lambda) \frac{\lambda^k}{k!}.
\end{align*}
This is denoted $N \sim \Poissonmad(\lambda)$. $N \sim \Poissonmad(0)$ and $N \sim \Poissonmad(\infty)$ are the random variables whose values are $0$ and $\infty$ with probability one. If $N \sim \Poissonmad(\lambda)$, then $\expectmad(N) = \lambda$.
\end{definition}
\noindent The Poisson distribution is particularly suited to modelling random counts, because it is countably additive in the rate. 
\begin{lemma}
\label{lem:po}
If $N_i \sim \Poissonmad(\lambda_i)$ independent with $\lambda_i \in \R^{\geq 0}$, then
\begin{align*}
\sum\nolimits_{i=1}^{\infty} N_i \sim \Poissonmad\left(\sum\nolimits_{i=1}^{\infty} \lambda_i\right).
\end{align*}
\end{lemma}
\begin{proof}

\citep{kingman1992poisson}. Let $S_m = \sum_{i=1}^m N_i$ and assume $\lambda_i > 0$ without loss of generality. Then for $S_2$,
\begin{align*}
\proba(S_2 = k) &= \sum_{r=0}^k \proba(N_1 = r, N_2 = k-r)\\
&= \sum_{r=0}^k \exp(-\lambda_1)\frac{\lambda_1^r}{r!} \exp(-\lambda_2)\frac{\lambda_2^{k-r}}{(k-r)!}\\
&=  \frac{\exp(-\lambda_1 - \lambda_2)}{k!} \sum_{r=0}^k {k \choose r}\lambda_1^r\lambda_2^{k-r}\\
&=  \frac{\exp(-\lambda_1 - \lambda_2)}{k!} (\lambda_1 + \lambda_2)^k.
\end{align*}
By induction Lemma \ref{lem:po} also holds for $S_m$. For infinite sums the events $\{S_m \leq k\}$ are nonincreasing. Thus,
\begin{align*}
\proba(S_{\infty} \leq k) &= \lim_{m \to \infty} \proba(S_m \leq k) = \sum_{j=1}^k \lim_{m \to \infty} \exp\left(- \sum\nolimits_{i=1}^m \lambda_i\right) \frac{(\sum\nolimits_{i=1}^m \lambda_i)^j}{j!}.
\end{align*}
\end{proof}
\noindent Because expectations distribute over infinite sums of positive random variables, the Poisson rate $\mu(B) = \expectmad(N(B))$ must also be a measure.

Instead of starting with a definition of Poisson processes, we work backwards from an algorithmic construction. \Algo{alg:pp} is a procedure that realizes a Poisson process $\Pi$ for a specified mean measure $\mu$. \Algo{alg:pp} iterates through a partition $\{B_i\}_{i=1}^{\infty}$ of $\R^n$. For each $B_i$ it first decides the number of points to place in $\Pi$ by sampling a Poisson with rate given by the measure, $N_i \sim \Poissonmad(\mu(B_i))$. Then, it places $N_i$ points by sampling independently from the probability distribution proportional to $\mu$ restricted to $B_i$. Normally, $X \sim \mathcal{D}$ is just a statement about the marginal distribution of $X$. In the context of an Algorithm box we also implicitly assume that it implies independence from all other random variables. We should note that \Algo{alg:pp} operates on volumes and samples from $\mu$. This is not an issue, if we think of it as a mathematical construction. It would be an issue, if we set out to simulate $\Pi$ on a computer.

\Algo{alg:pp} will occasionally have pathological behaviour, unless we restrict $\mu$ further. First, we require that each subset $B_i$ of the partition has finite measure; if $\mu(B_i) = \infty$, then \Algo{alg:pp} will stall when it reaches $B_i$ and fail to visit all of $\R^n$. If a partition $\{B_i\}_{i=1}^{\infty}$ with $\mu(B_i) < \infty$ exists for measure $\mu$, then $\mu$ is called $\sigma$-finite. Second, we want the resulting counts $N(B_i)$ to match the number of points placed $N_i$. This can be ensured if all of the points $X_{ij}$ are distinct with probability one. It is enough to require that $\mu(\{x\}) = 0$ for all singleton sets $x \in \R^n$. This kind of measure is known as nonatomic.

\begin{algorithm}[t]
\caption{A Poisson process $\Pi$ with $\sigma$-finite nonatomic mean measure $\mu$} \label{alg:pp}
\begin{algorithmic}
	\State Let $\{B_i\}_{i=1}^{\infty}$ be a partition of $\R^n$ with $\mu(B_i) < \infty$
	\State $\Pi = \emptyset$
	\For{$i=1$ to $\infty$}
		\State $N_i \sim \Poissonmad(\mu(B_i))$
		\For{$j=1$ to $N_i$}
			\State $X_{ij} \sim \mu(\cdot \cap B_i)/\mu(B_i)$
			\State $\Pi = \Pi \cup \{X_{ij}\}$
		\EndFor
	\EndFor
\end{algorithmic}
\end{algorithm}

The crucial property of the sets $\Pi$ produced by \Algo{alg:pp} is that the number of points $N(A_j)$ that fall in \emph{any} finite collection $\{A_j\}_{j=1}^m$ of disjoint sets are independent Poisson random variables. Clearly, the counts $N(B_i)$ for the partitioning sets of \Algo{alg:pp} are independent Poissons; it is not obvious that this is also true for other collections of disjoint sets. To show this we study the limiting behaviour of $N(B)$ by counting the points placed in $B_i \cap B$ and summing as \Algo{alg:pp} iterates over $\R^n$.
\begin{theorem}
\label{thm:pp}
Let $\Pi \subseteq \R^n$ be the subset realized by \Algo{alg:pp} with $\sigma$-finite nonatomic mean measure $\mu$ and $A_1, \ldots A_m \subseteq \R^n$ disjoint. $N(B) = \#(\Pi \cap B)$ for $B \subseteq \R^n$ satisfies 
\begin{enumerate}
\item $N(A_j) \sim \Poissonmad(\mu(A_j))$,
\item $N(A_j)$ are independent.
\end{enumerate}
\end{theorem}
\begin{proof}
Adapted from \cite{kingman1992poisson}. Let $B_i$ be the partition of \Algo{alg:pp} with $\mu(B_i) > 0$ without loss of generality. With probability one,
\begin{align*}
N(A_j) = N(\cup_{i=1}^{\infty} B_i \cap A_j) = \sum_{i=1}^{\infty} N(B_i \cap A_j).
\end{align*}
Consider the array of $N(B_i \cap A_j)$ for $i \in \{1, 2, \ldots\}$ and $j \in \{1, \ldots, m\}$. The rows are clearly independent. Thus, by Lemma \ref{lem:po} it is enough to show
\begin{enumerate}
\item $N(B_i \cap A_j) \sim \Poissonmad(\mu(B_i \cap A_j))$,
\item $N(B_i \cap A_j)$ for $j \in \{1, \ldots, m\}$ are independent,
\end{enumerate}
Let $A_0$ be the complement of $\cup_{i=1}^m A_i$. Because $\mu$ is nonatomic, each point is distinct with probability one. Thus,
\begin{align*}
\proba(N(B_i \cap A_0) = k_0, \ldots, N(B_i \cap A_m) = k_m &| N_i = k) =\\
&\frac{k!}{k_0 ! \ldots k_m!} \prod_{j=0}^m \frac{\mu(B_i \cap A_j)^{k_j}}{\mu(B_i)^{k_j}}
\end{align*} 
with $k_0 = k - \sum_{j=1}^m k_j$. Now,
\begin{align*}
\proba(N(B_i \cap A_1) &= k_1, \ldots, N(B_i \cap A_m) = k_m) =\\
&\sum_{k = \sum_j k_j}^{\infty} \exp(-\mu(B_i)) \frac{\mu(B_i)^k}{k!} \frac{k!}{k_0 ! \ldots k_m!} \prod_{j=0}^m \frac{\mu(B_i \cap A_j)^{k_j}}{\mu(B_i)^{k_j}}\\
&\sum_{k_0 = 0}^{\infty} \prod_{j=0}^m \exp(-\mu(B_i \cap A_j)) \frac{\mu(B_i \cap A_j)^{k_j}}{k_j!}\\
&= \prod_{j=1}^m \exp(-\mu(B_i \cap A_j)) \frac{\mu(B_i \cap A_j)^{k_j}}{k_j!}.
\end{align*}
finishes the proof.
\end{proof}

Notice that the partition in \Algo{alg:pp} has an indistinguishable effect on the eventual counts $N(B)$. In fact there may be entirely different algorithms that realize random subsets indistinguishable from $\Pi$. This motivates the standard definition for deciding whether a random process is Poisson.
\begin{definition}[Poisson process] 
\label{def:pp}
Let $\mu$ be a $\sigma$-finite nonatomic measure on $\R^n$. A random countable subset $\Pi \subseteq \R^n$ is a Poisson process with mean measure $\mu$ if
\begin{enumerate}
\item For $B \subseteq \R^n$, $N(B) \sim \Poissonmad(\mu(B))$.
\item For $A_1, \ldots A_m \subseteq \R^n$ disjoint, $N(A_j)$ are independent.
\end{enumerate}
\end{definition}
\noindent \Algo{alg:pp} together with Theorem \ref{thm:pp} is an existence proof for Poisson processes. Poisson processes are generic models for procedures that place points completely randomly in space. In later sections we specialize them to model the sequence of points considered by Monte Carlo methods.

\subsection{Mapping and thinning a Poisson process}

We are ultimately interested in understanding how the operations of accept-reject and the Gumbel-Max trick modify distributions. They are special cases of more generic operations on the points $X \in \Pi$ of a Poisson process, which modify its measure. Accept-reject corresponds to the stochastic removal of points based on their location. The Gumbel-Max trick corresponds to the deterministic relocation of points. Here we study those operations in some generality. 

The stochastic removal of points $X \in \Pi$ is called thinning. To count the number of points that remain after thinning, we need their joint distribution before thinning. If we restrict our attention to one of the subsets $B_i$ of the partition in \Algo{alg:pp}, then the distribution is clear: conditioned on $N(B_i) = k$, each point is distributed identically and independently (i.i.d.) as $\mu$ restricted to $B_i$. This property turns out to be true for any subset $B \subseteq \R^n$ of finite measure.
\begin{lemma} 
\label{lem:bernoulli}
Let $\Pi \subseteq \R^n$ be a Poisson Process with $\sigma$-finite nonatomic mean measure $\mu$ and $B \subseteq \R^n$ with $0 < \mu(B) < \infty$. Given $N(B) = k$, each $X_i \in \Pi \cap B$ for $i \in \{1, \ldots k\}$ is i.i.d. as,
\begin{displaymath}
X_i \, | \, \{N(B) = k\} \sim \mu(\cdot \cap B)/\mu(B).
\end{displaymath}
\end{lemma}
\begin{proof}
The proof is uninformative, so we leave it to the Appendix.
\end{proof}
\noindent Intuitively, this result ought to be true, because we could have realized $\Pi$ via \Algo{alg:pp} with $B$ as one of the partitioning sets. 

Now suppose we remove points $X \in \Pi$ independently with probability $1-\rho(X)$, where $\rho : \R^n \to [0, 1]$ is some integrable function. For $B$ with finite measure, given $N(B)$ the probability of keeping $X \in \Pi \cap B$ is
\begin{align}
\label{eq:keepprob}
\proba(\text{keep } X\, | \, N(B) = k) = \expectmad (\rho(X) \, | \, N(B) = k) = \int_{B} \frac{\rho(x)}{\mu(B)} \mu(dx).
\end{align}
By summing over the value of $N(B)$, we can derive the marginal distribution over the number of remaining points. This is the basic strategy of the Thinning Theorem.
\begin{theorem}[Thinning] 
\label{thm:thin}
Let $\Pi \subseteq \R^n$ be a Poisson Process with $\sigma$-finite nonatomic mean measure $\mu$ and $S(x) \sim \Bernoulli(\rho(x))$ an independent Bernoulli random variable for $x \in \R^n$ with integrable $\rho: \R^n \to [0,1]$, then
\begin{align}
\label{eq:thindef}
\thin(\Pi, S) = \{X : X \in \Pi \text{ and } S(X) = 1\}
\end{align}
is a Poisson process with mean measure
\begin{displaymath}
\mu^*(B) =  \int_{B} \rho(x) \mu(dx).
\end{displaymath}
\end{theorem}
\begin{proof}
Originally from \cite{thinning}. Let $B \subseteq \R^n$. Define,
\begin{align*}
N^*(B) = \#(\thin(\Pi, S) \cap B) 
\end{align*}
$N^*(B)$ clearly satisfies the independence property and the result is trivial for $\mu(B) = 0$. For $0 < \mu(B) < \infty$,
\begin{align*}
\proba(N^*(B) = k) &= \sum_{j=k}^{\infty} \proba(N(B) = j)\proba(k \text{ of } S(X_i) = 1 | N(B) = j).
\intertext{Let $\bar{\mu}^*(B) = \mu(B) - \mu^*(B)$. By (\ref{eq:keepprob}),}
&= \sum_{j=k}^{\infty} \exp(-\mu(B)) \frac{\mu(B)^j}{j!} {j \choose k} \frac{\mu^*(B)^k}{\mu(B)^k} \frac{\bar{\mu}^*(B)^{j-k}}{\mu(B)^{j-k}}\\
&=  \exp(-\mu^*(B))\frac{\mu^*(B)^k}{k!} \sum_{j=k}^{\infty} \exp(-\bar{\mu}^*(B)) \frac{\bar{\mu}^*(B)^{j-k}}{(j-k)!} \\
&=  \exp(-\mu^*(B))\frac{\mu^*(B)^k}{k!}.
\end{align*}
For $\mu(B) = \infty$, partition $B$ into subsets with finite measure. The countable additivity of integrals of nonnegative functions and of Poisson random variables (Lemma \ref{lem:po}) finishes the proof.
\end{proof}

A measurable function $h : \R^n \to \R^n$ that relocates points $X \in \Pi$ is easy to analyze if it is 1-1, because it will not relocate two distinct points to the same place. The key insight is that we can count the points relocated to $B \subseteq \R^n$ by counting in the preimage $h^{-1}(B)$; the so-called Mapping Theorem.
\begin{theorem}[Mapping] 
\label{thm:map}
Let $\Pi \subseteq \R^n$ be a Poisson process with $\sigma$-finite nonatomic mean measure $\mu$ and $h : \R^n \rightarrow \R^n$ a measurable 1-1 function, then 
\begin{align*}
h(\Pi) = \{h(X) : X \in \Pi\}
\end{align*}
is a Poisson process with mean measure
\begin{displaymath}
\mu^*(B) = \mu(h^{-1}(B))
\end{displaymath}
\end{theorem}
\begin{proof} Adapted from \cite{kingman1992poisson}. $h$ is 1-1, therefore
\begin{displaymath}
\# (\{h(X) : X \in \Pi\} \cap B) = \# \{X \in \Pi: X \in h^{-1}(B)\} \sim \Poissonmad(\mu(h^{-1}(B))).
\end{displaymath}
Pre-images preserve disjointness, so the independence property is guaranteed. 1-1 functions map partitions of the domain to partitions of the range, so $\mu^*$ is still $\sigma$-finite.
\end{proof}

\section{Exponential races}
\label{sec:er}

\subsection{Definition and first arrivals distribution} 
In this section we specialize the Poisson process to model the sequence of points considered by accept-reject and the Gumbel-Max trick. We call the model an exponential race as a reference to a classical example. An exponential race (occasionally race for short) is a Poisson process in $\R^+ \times \R^n$, which we interpret as points in $\R^n$ ordered by an arrival time in the positive reals $\R^+$. The ordered points of an exponential race have a particularly simple distribution; the location in $\R^n$ of each point is i.i.d. according to some arrival distribution and the rate at which points arrive in time depends stochastically on the normalization constant of that arrival distribution. The Thinning and Mapping Theorems of Poisson processes have corresponding lemmas for exponential races, which describe operations that modify the arrival distribution of an exponential race. The ultimate value of this model is that a variety of apparently disparate Monte Carlo methods can be interpreted as procedures that simulate an exponential race. In Section \ref{sec:alg} we present Monte Carlo methods which produce samples from intractable distributions by operating on the simulation of an exponential race with a tractable distribution.
In this section we define an exponential race for an arbitrary finite nonzero measure $P$, discuss strategies for simulating exponential races when $P$ is tractable, and derive two operations that modify the arrival distribution of exponential races.

\begin{figure}[t]
\begin{center}
\includegraphics{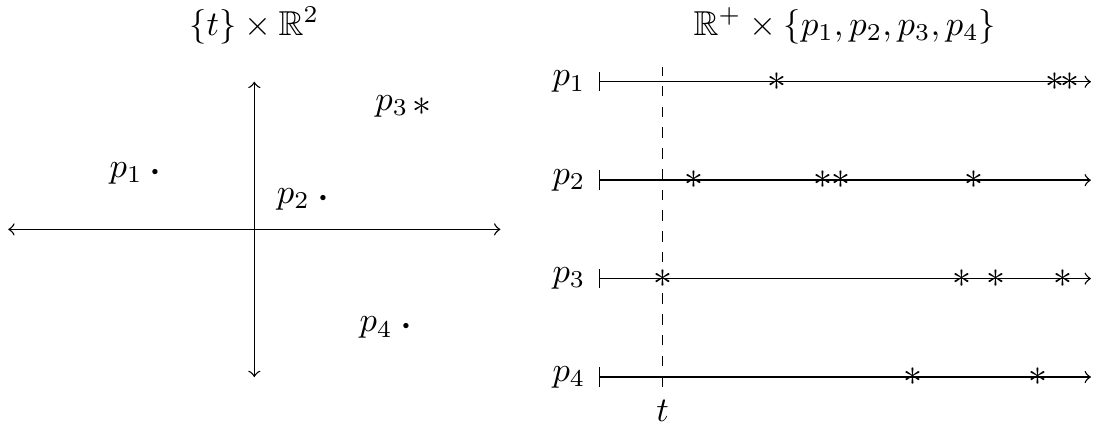}
\caption{The realization of an exponential race with points arriving at $p_j \in \R^2$. The left hand plot shows the location of arrivals in the plane $\R^2$ and the first arrival at time $t$ at $p_3$. The right hand plot shows future arrival times at the four points.}
\label{fig:er}
\end{center}
\end{figure}

For motivation we review the traditional exponential race example (see \citeauthor{durrett2012essentials}, \citeyear{durrett2012essentials}). Imagine instantaneous flashes of light arriving in time at $m$ distinct points $p_j$ scattered in $\R^2$. Suppose the arrival times of the flashes at each $p_j$ are determined by independent Poisson processes $\Pi_j \subseteq \R^+$ with mean measure $\lambda_j((0, t]) = \lambda_j t$  and $\lambda_j > 0$, see \Fig{fig:er}. The question is which point will get the first flash of light and how long do we need to wait? 
The first arrival at $p_j$ is after time $t$ iff $\Pi_j \cap (0, t]$ is empty,
\begin{align}
\label{eq:arrival}
\proba(T_j > t) = \proba(\# (\Pi_j \cap (0, t]) = 0) = \exp(-\lambda_j t).
\end{align}
(\ref{eq:arrival}) is the complementary cumulative distribution function of an exponential random variable, which we briefly review.
\begin{definition}[Exponential random variable]
$E$ is an exponential random variable distributed on positive $t \in \R^+$ with nonnegative rate $\lambda \in \R^{\geq 0}$ if
\begin{align}
\proba(E > t) = \exp(-\lambda t).
\end{align}
This is denoted $E \sim \Exponential(\lambda)$ and $E \sim \Exponential(0)$ is the random variable whose value is $\infty$ with probability one. If $E \sim \Exponential(1)$, then $E/\lambda \sim \Exponential(\lambda)$.
\end{definition}
\noindent Thus, the location and time of the first arrival is determined by the minimum of $m$ exponential random variables. For exponential random variables this is particularly easy to analyze; the minimum is an exponential random variable with rate $\sum_{j=1}^m \lambda_j$ and it is achieved at the $j$th variable with probability proportional to the rate $\lambda_j$. Surprisingly, these values are independent. 
\begin{lemma}
\label{lem:exp}
Let $E_j \sim \Exponential(\lambda_j)$ independent with nonegative $\lambda_j \in \R^{\geq 0}$. If
\begin{displaymath}
E^* = \min_{1 \leq j \leq m} E_j \text{ and } J^* = \argmin_{1 \leq j \leq m} E_j,
\end{displaymath}
and at least one $\lambda_j >0$ then
\begin{enumerate}
\item The density of $E_j$ with $\lambda_j > 0$ is $\lambda_j \exp(-\lambda_j t )$ for $t \in \R^+$,
\item $E^* \sim \Exponential(\sum\nolimits_{j=1}^m \lambda_j)$,
\item $\proba(J^* =k) \propto \lambda_k$,
\item $E^*$ is independent of $J^*$.
\end{enumerate}
\end{lemma}
\begin{proof}
\begin{enumerate}
\item The derivative of $1 - \exp(-\lambda_j t)$ is $\lambda_j \exp(-\lambda_j t )$.
\end{enumerate}
\noindent 2., 3., 4.  Note that with probability 1 the $E_j$ will be distinct, so
\begin{align*}
\proba(J^* =k, E^* > t) &= \proba(\cap_{j \neq k}\{E_j > E_k > t\})\\
&= \int_{t}^{\infty} \lambda_k \exp(-\lambda_k x) \prod\nolimits_{j \neq k} \exp(-\lambda_j x) \, dx\\
&= \frac{\lambda_k}{\sum\nolimits_{j=1}^m \lambda_j} \int_{t}^{\infty} (\sum\nolimits_{j=1}^m \lambda_j) \exp(-\sum\nolimits_{j=1}^m \lambda_j x) \, dx\\
&= \frac{\lambda_k}{\sum\nolimits_{j=1}^m \lambda_j} \exp(-\sum\nolimits_{j=1}^m \lambda_j t).
\end{align*}
This finishes the lemma.
\end{proof}

The extension of exponential races to arbitrary distributions on $\R^n$ is straightforward. The $m$ Poisson processes of the example are together a single Poisson process on $\R^+ \times \R^n$ with mean measure $(\lambda \times P)((0, t] \times B) = \sum_{j=1}^m t \lambda_j 1_{B}(p_j)$. $\lambda \times P$ is the product measure on $\R^+ \times \R^n$, where each is respectively equipped with $\lambda((0, t]) = t$ and $P(B) =  \sum_j \lambda_j 1_{B}(p_j)$. Extending this idea to an arbitrary finite measure $P$ (not just the discrete measures) is the key idea behind exponential races. Notice that $P$ in our example is atomic, which is fine, because the product measure $\lambda \times P$ is not atomic. On the other hand, we want the points arriving in $\R^n$ to correspond to the probability distribution $P(\cdot)/P(\R^n)$, so we will require that $P$ is finite, $P(\R^n) < \infty$, and nonzero, $0 < P(\R^n)$. Also, in contrast to Poisson processes, exponential races have a natural ordering in time.

\begin{definition}[Exponential race] Let $P$ be a finite nonzero measure on $\R^n$. A random countable subset $R \subseteq \R^+ \times \R^n$ is an exponential race with measure $P$ if the following hold
\begin{enumerate}
\item $R$ is a Poisson process with mean measure $\lambda \times P$.
\item $R$ is totally ordered by time, the first coordinate.
\end{enumerate}
If $R = \{(T_i, X_i)\}_{i=1}^{\infty}$, then we assume the enumeration corresponds to the ordering so that $i < j$ implies $T_i < T_j$.
\end{definition}

We can realize an exponential race with a slight modification of \Algo{alg:pp}; use the partition of rectangles $B_i = (i-1, i] \times \R^n$, and sort points by their time variable.

This is not the most direct characterization, so instead we derive the joint distribution of the first $m$ ordered points in Theorem \ref{thm:er}. The distribution of the countably infinite set $R$ is completely described by the joint distribution of the first $m$ points for all finite $m$. The proof of Theorem \ref{thm:er} shows that the locations $X_i$ are independently distributed as $P(\cdot)/P(\R^n)$ and the interarrival times $T_i - T_{i-1}$ are independent and exponentially distributed with rate $P(\R^n)$. This theorem is the cornerstone of this chapter, because it suggest a strategy for proving the correctness of Monte Carlo methods; if we can prove that the output of an algorithm $(T, X)$ is the first arrival of an exponential race with measure $P$, then Theorem \ref{thm:er} guarantees that the location $X$ is a sample from $P(\cdot)/P(\R^n)$.

\begin{theorem} 
\label{thm:er}
Let $P$ be a finite nonzero measure on $\R^n$, $X_i \sim P(\cdot)/P(\R^n)$ independent, and  $E_i \sim \Exponential(P(\R^n))$ independent, then first $m$ points $\{(T_i, X_i)\}_{i=1}^m$ of any exponential race $R \subseteq \R^+ \times \R^n$ with measure $P$ have the same joint distribution as
\begin{align*}
\{(\sum\nolimits_{j=1}^i E_j,  X_i)\}_{i=1}^m.
\end{align*}

\end{theorem}
\begin{proof} Let $T(t, B)$ be the time of the first arrival in $B$ after time $t \geq 0$,
\begin{align}
\label{eq:arrivaltime}
T(t, B) = \min \{T_i : (T_i, X_i) \in R \cap (t, \infty) \times B\}.
\end{align}
$R \cap ((t, s + t] \times B)$ is finite with probability one for all $s > 0$, so (\ref{eq:arrivaltime}) is well defined. $T(t, B) - t$ is an exponential random variable, because
\begin{align*}
\proba(T(t, B) - t > s) = \proba(N((t, s+t] \times B) = 0) = \exp(- P(B)s).
\end{align*}
$T(t, B)$ and $T(t, B^c)$ are independent, by Poisson process independence.

We proceed by induction. The event $\{T_1 > s, X_1 \in B\}$ is equivalent to $\{T(0, B^c) > T(0, B) > s\}$. $P(B) > 0$ or $P(B^c) > 0$, so by Lemma \ref{lem:exp},
\begin{align*}
\proba(T_1 > s, X_1 \! \in \! B) = \proba(T(0, B^c) > T(0, B) > s) = \exp(-s P(\R^n))\frac{ P(B)}{P(\R^n)}.
\end{align*}
Now, assume Theorem \ref{thm:er} holds for $k$. The event 
\begin{align*}
\{T_i = t_i, X_i = x_i\}_{i=1}^k
\end{align*}
is completely described by counts in $(0, t_k] \times \R^n$ and thus independent of 
\begin{align*}
\{T(t_k, B^c) > T(t_k, B) > s + t_k\}
\end{align*}
Thus
\begin{align*}
\proba(T_{k+1} - &T_k > s, X_{k+1} \! \in \! B | \{T_i = t_i, X_i = x_i\}_{i=1}^k) \\
&=\proba(T(t_k, B^c) > T(t_k, B) > s + t_k | \{T_i = t_i, X_i = x_i\}_{i=1}^k) \\
&=\proba(T(t_k, B^c) > T(t_k, B) > s + t_k) \\
&=  \exp(-s P(\R^n))\frac{ P(B)}{P(\R^n)}
\end{align*}
concludes the proof.\end{proof}

\subsection{Simulating an exponential race with a tractable measure}
\label{subsec:ersim}

\begin{algorithm}[b]
\caption{An exponential race $R$ with finite nonzero measure $Q$} \label{alg:er}
\begin{algorithmic}
	\State $R = \emptyset$
	\State $T_0 = 0$
	\For{$i=1$ to $\infty$}
		\State $E_i \sim \Exponential(Q(\R^n))$
		\State $X_{i} \sim Q(\cdot)/Q(\R^n)$
		\State $T_i = T_{i-1} + E_i$
		\State $R = R \cup \{(T_i, X_i)\}$
	\EndFor
\end{algorithmic}
\end{algorithm}

If $Q$ is a tractable finite nonzero measure on $\R^n$, that is we have a procedure for computing $Q(\R^n)$ and sampling from $Q(\cdot)/Q(\R^n)$, then Theorem \ref{thm:er} suggests \Algo{alg:er} for simulating an exponential race $R$ with measure $Q$.  \Algo{alg:er} simulates the points of an exponential race in order of arrival time. It does not terminate, but we can think of it as a coroutine or generator, which maintains state and returns the next arrival in $\R$ each time it is invoked. As a simple example consider the uniform measure $Q((a, b]) = b-a$ on $[0,1]$. \Algo{alg:er} for this $Q$ simulates a sequence of arrivals $\{(T_i, X_i)\}_{i=1}^{\infty}$ with arrival location $X_i \sim \Uniform[0,1]$ and interarrival time $T_{i+1} - T_{i} \sim \Exponential(1)$, see the left hand plot of Figure \ref{fig:eralt}.

As with the initial discrete example, in which we constructed an exponential race from $m$ independent Poisson processes, this is not the only approach. More generally, if $\{B_i\}_{i=1}^m$ is any finite partition of $\R^n$ such that $Q(\cdot \cap B_i)$ is tractable, then we can simulate $R$ by simulating $m$ independent exponential races $R_i$ with measure $Q(\cdot \cap B_i)/Q(B_i)$ via \Algo{alg:er} and sorting the result $\cup_{i=1}^m R_i$. This can be accomplished lazily and efficiently with a priority queue data type, which prioritizes the races $R_i$ according to which arrives next in time.
It also possible to split the races $R_i$ online by partitioning $B_i$ and respecting the constraint imposed by the arrivals already generated in $B_i$. We highlight a particularly important variant, which features in A* sampling in Section \ref{sec:alg}. Consider an infinitely deep tree in which each node is associated with a subset $B \subseteq \R^n$. If the root is $\R^n$ and the children of each node form a partition of the parent, then we call this a space partitioning tree. We can realize an exponential race over a space partitioning tree by recursively generating arrivals $(T, X)$ at each node $B$. Each location $X$ is sampled independently from $Q(\cdot \cap B)/Q(B)$, and each time $T$ is sampled by adding an independent $\Exponential(Q(B))$ to the parent's arrival time. The arrivals sorted by time over the realization of the tree form a exponential race. See \Fig{fig:eralt}.

\begin{figure}[t]
\begin{center}
\includegraphics{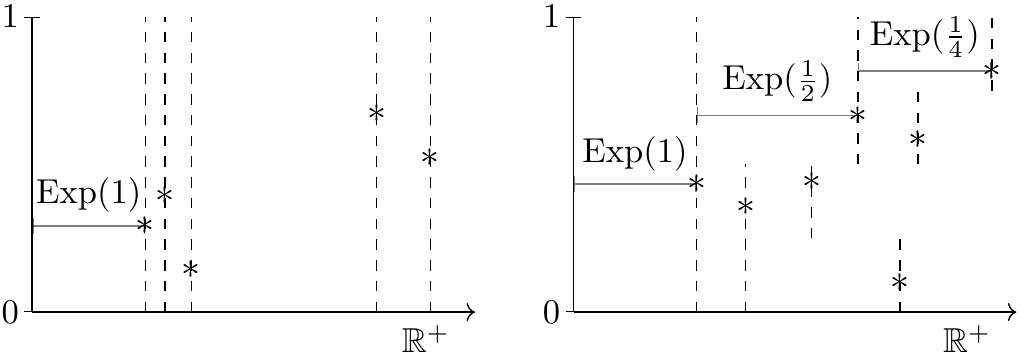}
\caption{Two methods for simulating an exponential race. The left hand plot shows the first arrivals of a uniform exponential race on $[0, 1]$ simulated by \Algo{alg:er}. The right hand plot shows the first arrivals of an exponential race simulated over a space partitioning tree. Dashed lines dominate the set in which an arrival is first.}
\label{fig:eralt}
\end{center}
\end{figure}

\subsection{Transforming an exponential race with accept-reject and perturb}
Most finite nonzero measures $P$ on $\R^n$ are not tractable. Monte Carlo methods accomplish their goal of sampling from intractable distributions by transforming samples of tractable distributions. In this subsection we present accept-reject and perturb operations, which transform a realization of an exponential race with measure $Q$ into a realization of an exponential race with a distinct measure $P$. In practice $Q$ will be tractable and $P$ intractable, so that simulating an exponential race with an intractable measure can be accomplished by simulating the points of an exponential race with a tractable measure, for example via \Algo{alg:er}, and transforming it with accept-reject or perturb operations. The accept-reject and perturb operations are named after their respective literatures, accept-reject corresponds to rejection sampling and perturb corresponds to the Gumbel-Max trick. The correspondence between the perturb operation and the Gumbel-Max trick may not be obvious, so we discuss this in Section \ref{sec:gup}.

Let $Q$ and $P$ be finite nonzero measures in $\R^n$. We assume that they have densities $g$ and $f$ with respect to some base measure $\mu$,
\begin{align}
\label{eq:density}
Q(B) = \int_B g(x) \mu(dx) \qquad P(B) = \int_B f(x) \mu(dx).
\end{align}
We assume that $g$ and $f$ have the same support and their ratio is bounded,
\begin{align}
\label{eq:bounded}
\mathrm{supp}(f) = \mathrm{supp}(g) \qquad \frac{f(x)}{g(x)} \leq M \text{ for all } x \in \mathrm{supp}(g)
\end{align}
where $\mathrm{supp}(g) = \{x \in \R^n : g(x) \neq 0\}$. The assumption $\mathrm{supp}(f) = \mathrm{supp}(g)$ can be softened here and throughout the chapter to $\mathrm{supp}(f) \subseteq \mathrm{supp}(g)$, but it complicates the analysis. The accept-reject strategy is to realize more points than needed from an exponential race with measure $MQ(\cdot)$ and stochastically \emph{reject} points with probability equal to the ratio of instantaneous rates of arrival, $f(x)/(g(x)M)$.
The perturbation strategy is to realize just the points needed from an exponential race with measure $Q$, but to \emph{perturb} the arrival times according to the transformation $t \to t g(x)/f(x)$ for all points arriving at $x$.

Before we present the proofs, consider the following intuition. Imagine taking a long exposure photograph of the plane as instantaneous flashes arrive according to an exponential race with measure $Q$. The rate at which points arrive will determine the intensity of a heat map with regions receiving more points brighter than those receiving fewer. Over time the relative intensities will correspond to the probability distribution proportional to $Q$. If someone were just ahead of us in time and stochastically discarded points that arrived in $B$ or delayed points in $B$ relative to points in $B^c$, then our perception of the likelihood of $B$ would change. Mired in time, we would not be able to distinguish whether points were discarded, reordered, or the true measure $Q$ was in fact different.

The correctness of these operations on an exponential race can be justified as special cases of the Thinning and Mapping Theorems.
\begin{lemma}[Accept-Reject]
\label{lem:accept}
Let $Q$ and $P$ be finite nonzero measures on $\R^n$ under assumptions (\ref{eq:density}) and (\ref{eq:bounded}). If $R\subseteq \R^+ \times \R^n$ is an exponential race with measure $MQ(\cdot)$ and $\keep(t,x) \sim \Bernoulli(\rho(t,x))$ is i.i.d. for all $(t, x)$ with probability \begin{align*}
\rho(t,x) = \frac{f(x)}{g(x)M},
\end{align*}
then $\thin(R, \keep)$, from (\ref{eq:thindef}), is an exponential race with measure $P$.
\end{lemma}
\begin{proof}
By the Thinning Theorem, the mean measure of $\thin(R, \keep)$ is
\begin{align*}
\iint\limits_{B} \frac{f(x) }{g(x)M} g(x)M\mu(dx)  \lambda(dt) = \iint\limits_{B} f(x) \mu(dx)  \lambda(dt) = (\lambda \times P)(B).
\end{align*}
for $B \subseteq \R^+ \times \mathrm{supp}(g)$. The subsampled $(T_i, X_i)$ are in order and thus an exponential race with measure $P$.
\end{proof}

\begin{lemma}[Perturbation]
\label{lem:perturb}
Let $Q$ and $P$ be finite nonzero measures on $\R^n$ under assumptions (\ref{eq:density}) and (\ref{eq:bounded}). If $R \subseteq \R^+ \times \R^n$ is an exponential race with measure $Q$ and 
\begin{align*}
\shear(t, x) = \left(t\frac{g(x)}{f(x)}, x\right),
\end{align*}
then $\sort(\shear(R))$ is an exponential race with measure $P$ where $\sort$ totally orders points by the first coordinate, time.
\end{lemma}
\begin{proof}
$\shear$ is 1-1 on $\mathrm{supp}(f)$, so the Mapping Theorem applies. It is enough to check the mean measure of $\shear(R)$ on subsets of the form $B = (0, s] \times A$ for $s \in \R^+$ and $A \subseteq \mathrm{supp}(g)$,
\begin{align*}
\iint\limits_{h^{-1}(B)} g(x) \lambda(dt) \mu(dx) = \int\limits_A  g(x) s\frac{f(x)}{g(x)} \mu(dx) = (\lambda \times P)(B).
\end{align*}
Thus, sorting $\shear(T_{i}, X_{i})$ forms an exponential race with measure $P$.
\end{proof}

\section{Gumbel processes}
\label{sec:gup}
\subsection{Definition and construction}
The central object of the Gumbel-Max trick is a random function over a finite set whose values are Gumbel distributed. Gumbel valued functions over a finite choice set are extensively studied in random choice theory, where there is a need for a statistical model of utility (\citeauthor{yellott1977relationship}, \citeyear{yellott1977relationship} for example). The extension to Gumbel valued functions over continuous spaces has been explored in random choice theory \citep{malmberg2013random} and in the context of Monte Carlo simulation \citep{maddison2014astarsamp}. Following \cite{maddison2014astarsamp} we will refer to this class of Gumbel valued functions on $\R^n$ as Gumbel processes. Gumbel processes underpin the recent interest in perturbation based Monte Carlo methods, because their maxima are located at samples from probability distributions, see also \citep{papandreou2011perturb, tarlow2012randomized, hazan2013perturb, 2015arXiv150609039C, kim2016lprelaxsamp}. In this section we clarify the connection between Gumbel processes and our development of exponential races. We will show that the value of a Gumbel process at $x \in \R^n$ can be seen as the log transformed time of the first arrival at $x$ of some exponential race. This has the advantage of simplifying their construction and connecting the literature on the Gumbel-Max trick to our discussion. Related constructions have also been considered in the study of extremal processes \citep{resnick2013extreme}. In this subsection we define and construct Gumbel processes. In the next subsection we discuss their simulation and present a generalized Gumbel-Max trick derived from the Perturbation Lemma.

The Gumbel distribution dates back to the statistical study of extrema and rare events \citep{gumbel}. The Gumbel is a member of a more general class of extreme value distributions. A central limit theorem exists for these distributions --- after proper renormalization the maximum of an i.i.d. sample of random variables converges to one of three possible extreme value distributions \citep{gedenko}. The Gumbel is parameterized by a location $\mu \in \R$.
\begin{definition}[Gumbel random variable] $G$ is a Gumbel distributed random variable on $\R$ with location $\mu \in \R$ if
\begin{align*}
\proba(G \leq g) = \exp(-\exp(-g + \mu))
\end{align*}
This is denoted $G \sim \Gumbel(\mu)$ and $G \sim \Gumbel(-\infty)$ is the random variable whose value is $-\infty$ with probability one. If $G \sim \Gumbel(0)$, then $G + \mu \sim \Gumbel(\mu)$.
\end{definition}
\noindent The Gumbel distribution has two important properties for our purposes. The distribution of the maximum of independent Gumbels is itself a Gumbel --- a property known as max-stability --- and the index of the maximum follows the Gibbs distribution: if $G(i) \sim \Gumbel( \mu_i)$, then 
\begin{align*}
\max_{1 \leq i \leq m} G(i) \sim \Gumbel(\log \sum\limits_{i=1}^m \exp(\mu_i)) \quad \argmax_{1 \leq i \leq m} G(i) \sim \frac{\exp(\mu_i)}{\sum_{i=1}^m \exp(\mu_i)}.
\end{align*}
 The Gumbel-Max trick of the introduction for sampling from a discrete distribution with mass function $f : \{1, \ldots, m\} \to \R^+$ is explained by taking $\mu_i = \log f(i)$. It is informative to understand these properties through the Gumbel's connection to the exponential distribution.
\begin{lemma}
\label{lem:gumbel}
If $E \sim \Exponential(\lambda)$ with nonnegative rate $\lambda \in \R^{\geq 0}$, then 
\begin{align*}
-\log E \sim \Gumbel(\log \lambda).
\end{align*}
\end{lemma}
\begin{proof}
$\proba(-\log E \leq g) = \proba(E \geq \exp(-g)) = \exp(- \exp(-g + \log \lambda))$
\end{proof}
\noindent Therefore the distribution of the maximum and argmaximum of Gumbels is explained by Lemma \ref{lem:exp}, because passing a maximization through $-\log$ becomes a minimization.

A Gumbel process $G : \R^n \to \R \cup \{-\infty\}$ is a Gumbel valued random function. Their characterizing property is that the maximal values of a Gumbel process over the subsets $B \subseteq \R^n$ are marginally Gumbel distributed with a location that scales logarithmically with the volume of $B$ according to some finite nonzero measure $P$,
\begin{align*}
\max_{x \in B} G(x) \sim \Gumbel(\log P(B))
\end{align*}
Implicit in this claim is the assertion that the maximizations $\max_{x \in B} G(x)$ are well-defined --- the maximum exists --- for all $B \subseteq \R^n$.

\begin{definition}[Gumbel process] 
\label{def:gup}
Let $P$ be a finite nonzero measure on $\R^n$, $G: \R^n \to \R\cup\{-\infty\}$ a random function, and 
\begin{align}
\label{eq:gupmax}
G^*(B) = \max_{x \in B} G(x).
\end{align}
$G$ is a Gumbel process with measure $P$ if
\begin{enumerate}
\item For $B \subseteq \R^n$, $G^*(B) \sim \Gumbel(\log P(B))$.
\item For $A_1, \ldots, A_m$ are disjoint, $G^*(A_i)$ are independent.
\end{enumerate}
\end{definition}
\noindent Note, the event that $\argmax_{x \in \R^n} G(x)$ lands in $B \subseteq \R^n$ depends on which of $G^*(B)$ or $G^*(B^c)$ is larger. Following this reasoning one can show that the argmax over $\R^n$ is distributed as $P(\cdot)/P(\R^n)$. 

The study of Gumbel processes can proceed without reference to exponential races, as in \cite{maddison2014astarsamp}, but our construction from exponential races is a convenient shortcut that allows us to import results from Section \ref{sec:er}. Consider the function that reports the arrival time of the first arrival at $x \in \R^n$ for an exponential race $R$ with measure $P$,
\begin{align*}
T(x) = \min \{T_i : (T_i, x) \in R\}
\end{align*}
This function is almost surely infinite at all $x$, but for any realization of $R$ it will take on finite value at countably many points in $\R^n$. Moreover, the minimum of $T(x)$ over subsets $B \subseteq \R^n$ is well-defined and finite for sets with positive measure $P(B) > 0$; it is exponentially distributed with rate $P(B)$. In this way we can see that $- \log T(x)$ is Gumbel process, Figure \ref{fig:gup}. 

\begin{figure}[t]
\begin{center}
\includegraphics{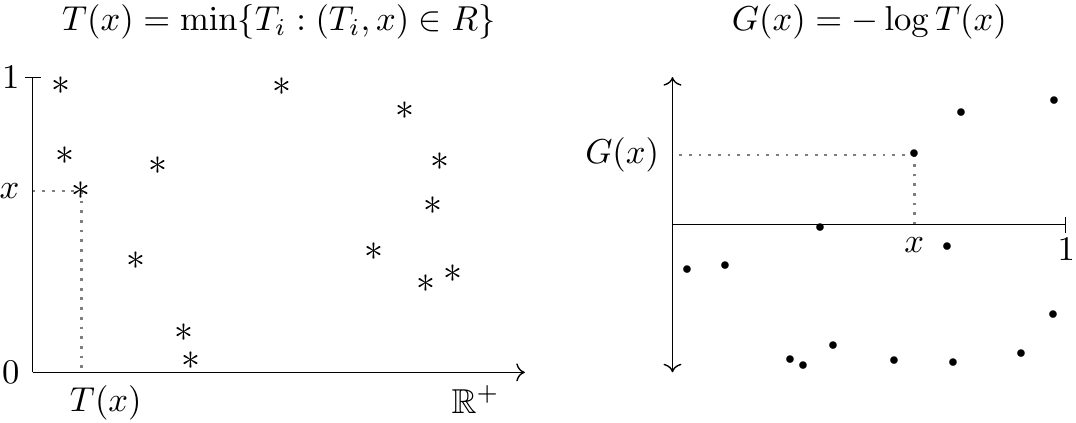}
\caption{Constructing a uniform Gumbel process $G : \R^n \to \R \cup \{-\infty\}$ on $[0,1]$ with an exponential race. The left hand plot shows the first arrivals $\ast$ of a uniform exponential race $R$. The right hand plot shows $G(x)$ set to $-\log$ the time $T(x)$ of the first arrival at $x$. The graph of $G(x)$ extends downwards to $-\infty$ taking on finite value at all points in $[0, 1]$ that have arrivals and $-\infty$ for all points with no arrivals.}
\label{fig:gup}
\end{center}
\end{figure}

\begin{theorem}
Let $R \subseteq \R^+ \times \R^n$ be an exponential race with measure $P$. 
\begin{align}
\label{eq:gupcon}
G(x) = -\log \min \{T_i : (T_i, x) \in R\}
\end{align}
is a Gumbel process with measure $P$.
\end{theorem}
\begin{proof}
First, for $x\in \R^n$
\begin{align*}
\min \{T_i : (T_i, x) \in R\} = T(0, \{x\}),
\end{align*}
where $T(0, B)$ is the first arrival time in subset $B \subseteq \R^n$ defined in (\ref{eq:arrivaltime}) from Theorem \ref{thm:er}. Thus $G^*(B)$ of (\ref{eq:gupmax}) is well defined, because
\begin{align*}
G^*(B) = \max_{x \in B} - \log \min \{T_i : (T_i, x) \in R\} = -\log T(0, B).
\end{align*}
$G^*(B)$ inherits the independence properties from Poisson process independence. Finally, Lemma \ref{lem:gumbel} gives us the marginal distribution of $G^*(B)$.
\end{proof}

\subsection{Simulating a Gumbel process and the Gumbel-Max trick}

Gumbel processes are relevant to Monte Carlo simulation in the same sense that we motivated exponential races --- if we can simulate the maximum value of a Gumbel process with measure $P$, then its location is a sample from the distribution $P(\cdot)/P(\R^n)$. \cite{maddison2014astarsamp} gave an algorithm for simulating Gumbel processes with tractable measures and a generalized Gumbel-Max trick for transforming their measure. We present those results derived from our results for exponential races.

\begin{algorithm}[b]
\caption{A Gumbel process with finite measure $Q$}\label{alg:gup}
\begin{algorithmic}
	\State Initialize $G(x) = -\infty$ for all $x \in \R^n$.
	\State $(\Omega_1, G_0, i)  = (\R^n, \infty, 1)$
	\While{$Q(\Omega_i) > 0$}
		\State $G_i \sim \TruncGumbel(\log Q(\Omega_i), G_{i-1})$
		\State $X_i \sim Q(\cdot \cap \Omega_i)/Q(\Omega_i)$
		\State $G(X_i) = G_i$ \% assign $G(x)$ at $X_i$ to $G_i$
		\State $\Omega_{i+1} = \Omega_{i} - \{X_i\}$
		\State $i = i + 1$
	\EndWhile
\end{algorithmic}
\end{algorithm}

The Gumbel process $G$ from construction (\ref{eq:gupcon}) has value $-\infty$ everywhere except at the countably many arrival locations of an exponential race. Therefore, for tractable measures $Q$ we could adapt \Algo{alg:er} for exponential races to simulate $G(x)$. The idea is to initialize $G(x) = -\infty$ everywhere and iterate through the points $(T_i, X_i)$ of an exponential race $R$ setting $G(X_i) = - \log T_i$. To avoid reassigning values of $G(x)$ we refine space as in Section \ref{subsec:ersim} by removing the locations generated so far. \Algo{alg:gup} implements this procedure, although it is superficially different from our description. In particular the value $G(X_i)$ is instead set to a truncated Gumbel $G_i \sim \TruncGumbel(\log Q(\Omega_i), G_{i-1})$, a Gumbel random variable with location $\log Q(\Omega_i)$ whose domain is truncated to $(-\infty, G_{i-1}]$. The connection to \Algo{alg:er} can be derived by decomposing the arrival times $T_i = \sum_{j=1}^i E_j$ for $E_j \sim \Exponential(Q(\Omega_j))$ and then considering the joint distribution of $G_i = - \log(\sum_{j=1}^i E_j)$. A bit of algebraic manipulation will reveal that
\begin{align*}
G_i \, | \, G_{i-1} \sim \TruncGumbel(\log Q(\Omega_i), G_{i-1})
\end{align*}
Thus, translating between procedures for simulating Gumbel processes and procedures for simulating exponential races is as simple as replacing chains of truncated Gumbels with partial sums of exponentials.

For continuous measures removing countably many points from the sample space has no effect, and in practice the removal line of \Algo{alg:gup} can be omitted. For those and many other measures \Algo{alg:gup} will not terminate; instead it iterates through the infinitely many finite values of $G(x)$ in order of their rank. For discrete measures with finite support \Algo{alg:gup} will terminate once every atom has been assigned a value. 

Finally, for simulating Gumbel processes with intractable measures $P$ the Perturbation Lemma of exponential races justifies a generalized Gumbel-Max trick. The basic insight is that multiplication by the ratio of densities $g(x)/f(x)$ becomes addition in log space.

\begin{lemma}[Gumbel-Max trick]
\label{lem:gumbelmax}
Let $Q$ and $P$ be finite nonzero measures on $\R^n$ with densities $g$ and $f$ under assumptions (\ref{eq:density}) and (\ref{eq:bounded}). If $G : \R^n \to \R \cap \{-\infty\}$ is a Gumbel process with measure $Q$, then
\begin{align*}
G^{\prime}(x) = \begin{cases}
\log f(x) - \log g(x) + G(x) & x \in \mathrm{supp}(g)\\
- \infty & \text{ otherwise}
\end{cases}
\end{align*}
is a Gumbel process with measure $P$. In particular for $G^* = \max_{x \in \R^n} G^{\prime}(x)$ and $X^* = \argmax_{x \in \R^n} G^{\prime}(x)$,
\begin{align*}
G^* \sim \Gumbel(\log P(\R^n)) \qquad X^* \sim P(\cdot)/P(\R^n)
\end{align*}
\end{lemma}
\begin{proof}
Arguing informally, this follows from the Perturbation Lemma applied to our construction (\ref{eq:gupcon}) of Gumbel processes. For $x \in \mathrm{supp}(g)$
\begin{align*}
\log f(x) - \log g(x) + G(x) = - \log \min \{T_i g(x)/ f(x) : (T_i, x) \in R\}.
\end{align*}
See \cite{maddison2014astarsamp} for a formal proof.
\end{proof}

\begin{figure}[t]
\begin{center}
\includegraphics{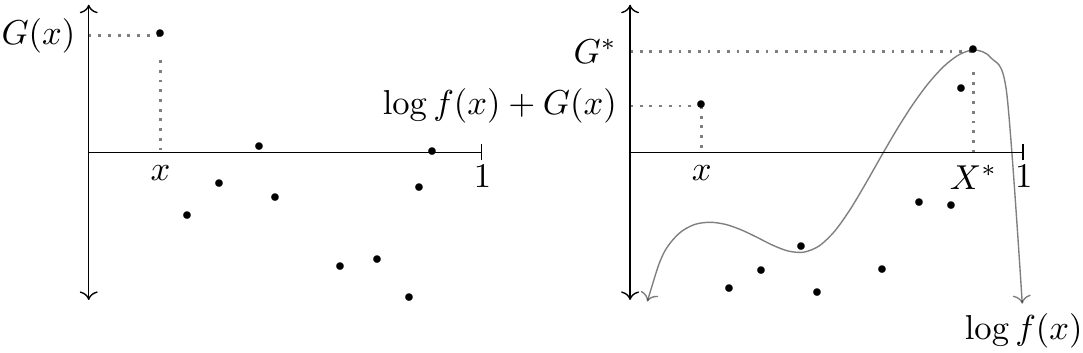}
\caption{A continuous Gumbel-Max trick. The left hand plot shows the maximal values of a uniform Gumbel process $G(x)$ on $[0,1]$. The right hand plot shows the result of perturbing $\log f(x)$ with $G(x)$. Notice that the ordering of values changes, and $X^*$ is now the location of the maximum $G^* = \max_x \log f(x) + G(x)$. Therefore, $X^*$ is a sample from the distribution with density proportional to $f(x)$.}
\label{fig:contgumbelmax}
\end{center}
\end{figure}
\noindent When $Q$ is the counting measure on $\{1, \ldots, m\}$, Lemma \ref{lem:gumbelmax} exactly describes the Gumbel-Max trick of the introduction. This brings full circle the connection between accept-reject and the Gumbel-Max trick. 

A Gumbel process is not profoundly different from an exponential race, but the difference of perspective --- a function as opposed to a random set --- can be valuable. In particular consider the following generalization of a result from Hazan and Jaakkola of this book. Let $G : \R^n \to \R \cup \{-\infty\}$ be a Gumbel process with measure $P$ whose density with respect to $\mu$ is $f$. If $G^* = \max_{x \in \R^n} G(x)$ and $X^*= \argmax_{x \in \R^n} G(x)$, then
\begin{align*}
\expectmad(G^*) = \log P(\R^n) + \gamma \qquad \expectmad(- \log f(X^*) + G^*) = H(f) + \gamma,
\end{align*}
where $H(f)$ is the entropy of a probability distribution with probability density function proportional to $f$ and $\gamma$ is the Euler-Mascheroni constant. Therefore the representation of probability distributions through Gumbel processes gives rise to a satisfying and compact representation of some of their important constants.

\section{Monte Carlo methods that use bounds}
\label{sec:alg}

\subsection{Rejection sampling}
In this section we present practical Monte Carlo methods that use bounds on the ratio of densities to produce samples from intractable distributions. We show how these methods can be interpreted as algorithms that simulate the first arrival of an exponential race. The basic strategy for proving their correctness is to argue that they perform accept-reject or perturb operations on the realization of an exponential race until they have provably produced the first arrival of the transformed race. We start by discussing the traditional rejection sampling and a related perturbation based method. Then we study OS* \citep{dymetman2012osstar}, an accept-reject method, and A* sampling \citep{maddison2014astarsamp}, a perturbation method. These algorithms have all been introduced elsewhere in the literature, so for more information we refer readers to the original papers.

Throughout this section our goal is to draw a sample from the probability distribution proportional to some measure $P$ with density $f$ with respect to some base measure $\mu$. We assume, as in the Accept-Reject and Perturbation Lemmas, access to a tractable proposal distribution proportional to a measure $Q$ with density $g$ with respect to $\mu$ such that $f$ and $g$ have the same support and the ratio $f(x)/g(x)$ is bounded by some constant $M$. For example consider the sample space $\{0,1\}^n$ whose elements are bit vectors of length $n$. A proposal distribution might be proportional to the counting measure $Q$, which counts the number of configurations in a subset $B \subseteq \{0,1\}^n$. Sampling from $Q(\cdot)/Q(\{0,1\}^n)$ is as simple as sampling $n$ independent $\Bernoulli(1/2)$.

Rejection sampling is the classic Monte Carlo method that uses bound information. It proposes $(X, U)$ from $Q$ and $\Uniform[0,1]$, respectively, and accepts $X$ if $U \leq f(X)/(g(X)M)$. The algorithm terminates at the first acceptance and is normally justified by noticing that it samples uniformly from the region under the graph of $f(x)$ by rejecting points that fall between $g(x)M$ and $f(x)$, see the left hand graph on \Fig{fig:alg} for an intuition. The acceptance decision also corresponds exactly to the accept-reject operation on exponential races, so we can interpret it as an procedure on the points of an exponential race. We call this procedure $\rej$ for short,
\begin{algorithmic}
	\For{$(T_i, X_i) \in R$ simulated by \Algo{alg:er} with measure $M Q(\cdot)$}
		\State $U_i \sim \Uniform[0,1]$.
		\If{$U_i < f(X_i)/(g(X_i)M)$}
			\Return $(T_{i}, X_i)$
		\EndIf
	\EndFor
\end{algorithmic}
The Accept-Reject Lemma guarantees that the returned values $(T, X)$ will be the first arrival of an exponential race with measure $P$, and Theorem \ref{thm:er} guarantees that $X$ is a sample from $P(\cdot)/P(\R^n)$. This is the basic flavour of the arguments of this section.

\begin{figure}[t]
\begin{center}
\includegraphics{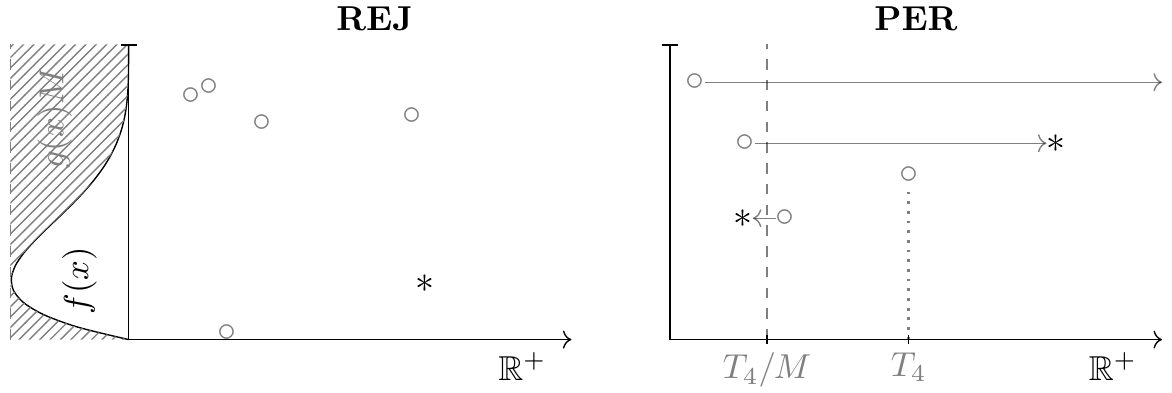}
\caption{Algorithms $\rej$ and $\per$ for measure $P$ on $[0,1]$ with proposal measure $Q$. The densities of $Q$ and $P$ are shown on the left hand side as densities over $x \in [0, 1]$. $\circ$ are arrivals of the race with measure $Q$, $\ast$ of the race with measure $P$. Both plots show the proposals considered until the first acceptance. For $\per$ opaque solid lines represent the perturb operation. $T_4$ is the fourth arrival from the race with measure $Q$. $T_4/M$ is the lower bound on all future arrivals, and thus all $\ast$ points to the left of $T_4/M$ are in order.}
\label{fig:alg}
\end{center}
\end{figure}

The Perturbation Lemma has a corresponding procedure, which uses the bound $M$ to provably return the first arrival of a perturbed exponential race. It is shown on the right hand side of \Fig{fig:alg}, and we call it $\per$.

\begin{algorithmic}
	\State $(T^*, X^*) = (\infty, \mathrm{null})$
	\For{$(T_i, X_i) \in R$ simulated by \Algo{alg:er} with measure $Q$}
		\If{$T^* > T_i g(X_i)/f(X_i)$}
			\State $T^* = T_ig(X_i)/f(X_i)$
			\State $X^* = X_i$
		\EndIf
		\If{$T_{i+1}/M \geq T^*$}
			\Return $(T^* , X^*)$
		\EndIf
	\EndFor
\end{algorithmic}

In this procedure $(T_i, X_i)$ iterates in order through the arrivals of an exponential race with measure $Q$. The perturbed times $T_ig(X_i)/f(X_i)$ will form a race with measure $P$, but not necessarily in order. $(T^*, X^*)$ are variables that track the earliest perturbed arrival so far, so $T^*$ is an upper bound on the eventual first arrival time for the race with measure $P$. $T_{i+1}$ is the arrival time of the next point in the race with measure $Q$ and $M$ bounds the contribution of the perturbation, so $T_{i+1}/M$ is a lower bound on the remaining perturbed arrivals. When $T^*$ and $T_{i+1}/M$ cross, $(T^*, X^*)$ is guaranteed to be the first arrival of the perturbed race.

$\rej$ and $\per$ can turned into generators for iterating through all of the arrivals of an exponential race with measure $P$ as opposed to just returning the first. For $\rej$ it is as simple as replacing {\bf return} with {\bf yield}, so that each time the generator is invoked it searches until the next acceptance and returns. For $\per$ we must store every perturbed arrival until its eventual order in the race with measure $P$ is determined. This can be accomplished with a priority queue $\Uc$, which prioritizes by earliest arrival time,
\begin{algorithmic}
	\State $\Uc = \mathrm{minPriorityQueue}()$
	\For{$(T_i, X_i) \in R$ simulated by \Algo{alg:er} with measure $Q$}
		\State $\Uc.\mathrm{pushWithPriortiy}(T_i g(X_i)/f(X_i), X_i)$
		\If{$T_{i+1}/M \geq \min \Uc$}
			{\bf yield} $\Uc.\mathrm{pop}()$
		\EndIf
	\EndFor
\end{algorithmic}
$\Uc$ takes the place of $T^*$ and $X^*$ in $\per$. The highest priority arrival on $\Uc$ will be the earliest of the unordered perturbed arrivals and $T_{i+1}/M$ is a lower bound on all future perturbed arrivals. When $T_{i+1}/M \geq \min \Uc$, the earliest arrival on $\Uc$ is guaranteed to be the next arrival. It is informative to think of the generator version of $\per$ via Figure \ref{fig:alg}. The lower bound $T_{i+1}/M$ is a bound across space that advances rightward in time, every arrival to the left of $T_{i+1}/M$ is in order and every arrival to the right is unordered.

Consider the number of iterations until the first acceptance in $\rej$ and $\per$. At first it seems that both algorithms should have different runtimes. $\rej$ is obviously memoryless, and it seems wasteful --- no information accumulates. On the other hand $\per$ accumulates the earliest arrival and its termination condition depends on a history of arrivals. Unfortunately, both algorithms have the same geometric distribution over the number of arrivals considered. Arguing informally, the lower bound $T_{i+1}/M$ of $\per$ plotted over the iterations will form a line with slope $(MQ(\R^n))^{-1}$. $\per$ terminates when this line crosses the first arrival time of the perturbed race. The first arrival of a race with measure $P$ occurs at $P(\R^n)^{-1}$ in expectation, so we expect the crossing point to occur on average at $MQ(\R^n)/P(\R^n)$ iterations. This is the same as the expected runtime of $\rej$.

\begin{lemma}
\label{lem:constantbound}
Let $K(\rej)$ and $K(\per)$ be the number of proposals considered by the rejection and perturbation sampling algorithms. Then
\begin{align*}
\proba(K(\rej) > k) = \proba(K(\per)> k) =  (1 - \rho)^{k} \text{ with } \rho = \frac{P(\R^n)}{Q(\R^n)M}.
\end{align*}
Thus $K(\rej)$ and $K(\per)$ are geometric random variable with
\begin{align*}
\expectmad(K(\rej)) = \expectmad(K(\per)) = \frac{1}{\rho}
\end{align*}
\end{lemma}
\begin{proof}
The probability of accepting a proposal at any iteration of $\rej$ is
\begin{align*}
\expectmad(f(X_i)/(g(X_i)M)) = \int \frac{f(x)}{g(x)M} \frac{g(x)}{Q(\R^n)} \mu(dx) =  \rho.
\end{align*}
Each decision is independent, so the probability of $k$ rejections is $(1 - \rho)^k$.

$\per$ exceeds $k$ iterations if $T_i g(X_i)/f(X_i) > T_{k+1}/ M$ for all $i \leq k$. Because the $X_i$ are i.i.d.,
\begin{align*}
\proba(K(\per)> k \ &| \  \{T_i = t_i\}_{i=1}^{k+1}) = \prod_{i=1}^k  \proba(t_i/t_{k+1} > f(X)/(g(X)M)),
\end{align*}
where $X\sim Q(\cdot)/Q(\R^n)$. Given $T_{k+1} = t_{k+1}$ the $T_i$ for $i \leq k$ are i.i.d. $T_i \sim \Uniform(0, t_{k+1})$ by Lemma \ref{lem:bernoulli}. Thus $T_i/T_{k+1} \sim \Uniform(0, 1)$ i.i.d.
\begin{align*}
\proba(K(\per)> k) = \prod_{i=1}^k  \proba(U > f(X)/(g(X)M)) = (1- \rho)^k
\end{align*}
finishes the proof.
\end{proof}

\subsection{Adaptive bounds}
\label{subsec:adaptivebounds}
Lemma \ref{lem:constantbound} is disappointing, because it suggests that reasoning about perturbations is as inefficient as discarding proposals. The problem is fundamentally that information carried in the bound $M$ about the discrepancy between $g(x)$ and $f(x)$ is static throughout the execution of both algorithms. Considering a contrived scenario will illustrate this point. Suppose that for every failed proposal we are given a tighter bound $M_{i+1} < M_{i}$ from some oracle. Both $\rej$ and $\per$ can be adapted to take advantage of these adaptive bounds simply by dropping in $M_i$ wherever $M$ appears.

In this case $\per$ is distinguished from $\rej$. $\rej$ makes an irrevocable decision at each iteration. In contrast $\per$ simply pushes up the lower bound $T_{i+1}/M_i$ without erasing its memory, bringing it closer to accepting the earliest arrival so far. Indeed, the probability of this oracle rejection sampling exceeding $k$ proposals is
\begin{align*}
\proba(K(\mathbf{OREJ}) > k) = \prod_{i}^k (1 - \rho_i) \text{ where } \rho_i = P(\R^n)/(Q(R^n)M_i).
\end{align*}
On the other hand, the probability of this oracle perturbation sampling exceeding $k$ proposals is
\begin{align*}
\proba(K(\mathbf{OPER}) > k) = \prod_{i=1}^k  \proba(U > f(X)/(g(X)M_k)) = (1 - \rho_k)^k,
\end{align*}
or the probability of rejecting $k$ proposals \emph{as if} the $M_k$th bound was known all along. By tracking the earliest arrival so far $\mathbf{OPER}$ makes efficient use of adaptive bound information, reevaluating all points in constant time.

\subsection{OS* adaptive rejection sampling and A* sampling}

The difference between $\rej$ and $\per$ exposed by considering adaptive bounds motivates studying OS* and A* sampling, Monte Carlo methods that use realistic adaptive bounds. Both methods iteratively refine a partition $\{B_i\}_{i=1}^m$ of $\R^n$, which allows them to use regional bounds $M(B_i)$, where $f(x)/g(x) \leq M(B_i)$ for $x \in B_i$. As with $\rej$ and $\per$, OS* and A* sampling are only distinguished by how they use this information. OS* reasons about accept-reject operations, A* sampling about perturb operations. In contrast to the relationship between $\rej$ and $\per$, A* sampling makes more efficient use of proposal samples than OS*.

OS* and A* sampling must compute volumes and samples of subsets under the proposal measure $Q$. It will be possibly intractable to consider any possible $B_i \subseteq \R^n$, so a user must implicitly specify a nice family $\Fc$ of subsets that is closed under a user-specified refinement function $\mathrm{split}(B, x)$. Hyperrectangles are a simple example. All together, the user must provide,
\begin{enumerate}
\item finite nonzero measure $P$ with a method for computing the density $f(x)$.
\item finite nonzero proposal measure $Q$ with methods for sampling restricted to $B \in \Fc$, computing measures of $B \in \Fc$, and computing the density $g(x)$.
\item partitioning set function $\mathrm{split}(B, x) \subseteq \Fc$ for $B \in \Fc$ that partitions $B$.
\item bounding set function $M(B)$ for $B \in \Fc$, $f(x)/g(x) \leq M(B)$ for $x \in B$.
\end{enumerate}
Specific examples, which correspond to experimental examples, are given in the Appendix.

\begin{algorithm}[t]
\caption{$\osstar$ adaptive rejection sampling for $P$ with proposal $Q$} \label{alg:osstar}
\begin{algorithmic}
	\State $\Pc_0 = \{\R^n\}$
	\State $T_0 = 0$
	\For{$i=1$ to $\infty$}
		\State $B_i \sim \proba(B) \propto Q(B)M(B)$ for $B \in \Pc_{i-1}$
		\State $X_{i} \sim Q(\cdot \cap B_i)/Q(B_i)$
		\State $E \sim \Exponential(\sum_{B \in \Pc_{i-1}} M(B)Q(B))$
		\State $T_i = T_{i-1} + E$
		\State $U_i \sim \Uniform[0,1]$
		\If{$U_i < f(X_i)/(g(X_i)M(B_i))$}
			\State\Return $(T_{i}, X_i)$
		\Else
			\State $\Cc = \mathrm{split}(B_i, X_i)$
			\State $\Pc_{i}  = \Pc_{i-1} - \{B_i\} + \Cc$
		\EndIf
	\EndFor
\end{algorithmic}
\end{algorithm}

OS* ($\osstar$ for short) is in a family of adaptive rejection sampling algorithms, which use the history of rejected proposals to tighten the gap between the proposal density and the density of interest. The name adaptive rejection sampling (ARS) is normally reserved for a variant that assumes $\log f(x)$ is concave \citep{gilks1992adaptive}. Accept-reject decisions are independent, so any adaptive scheme is valid as long as the rejection rate is not growing too quickly \citep{casella}. Our proof of the correctness appeals to exponential races, and it works for a wider range of adaptive schemes than just $\osstar$.

In more detail, $\osstar$ begins with the proposal density $g(x)$ and a partition $\Pc_{0} = \{\R^n\}$. At every iteration it samples from the distribution with density proportional to $\sum_{B \in \Pc_{i-1}} g(x)M(B)1_{B}(x)$ in a two step procedure, sampling a subset $B \in \Pc_{i-1}$ with probability proportional to $Q(B)M(B)$, and then sampling a proposal point $X$ from the distribution with density $g(x)$ restricted to $B$. If $X$ is rejected under the current proposal, then $P_{i-1}$ is refined by splitting $B$ with the user specified $\mathrm{split}(B, X)$. There is a choice of when to refine and which subset $B \in \Pc_{i-1}$ to refine, but for simplicity we consider just the form the splits the subset of the current proposal. $\osstar$ continues until the first acceptance, see \Algo{alg:osstar}.

\begin{theorem}[Correctness of OS*]
\label{thm:osstar}
Let $K(\osstar)$ be the number of proposal samples considered before termination. Then
\begin{align*}
\proba(K(\osstar) > k) \leq (1-\rho)^k \text{ where } \rho = \frac{P(\R^n)}{Q(\R^n)M(\R^n)}
\end{align*}
and upon termination the return values $(T, X)$ of OS* are independent and
\begin{align*}
T \sim \Exponential(P(\R^n)) \quad X \sim \frac{P(\cdot)}{P(\R^n)}.
\end{align*}
\end{theorem}
\begin{proof}
The situation is complicated, because the proposals $\{(T_i, X_i)\}_{i=1}^{\infty}$ of $\osstar$ are not an exponential race. Instead, we present an informal argument derived from a more general thinning theorem, Proposition 14.7.I. in \cite{daley2007introduction}. Let $g_i(x)$ be the proposal density at iteration $i$,
\begin{align*}
g_i(x) = \sum\nolimits_{B \in \Pc_{i-1}} g(x)M(B)1_B(x).
\end{align*}
Clearly, $g_i(x)$ depends on the history of proposals so far and $f(x) \leq g_i(x) \leq g(x)M(\R^n)$ for all $i$. Let $R$ be an exponential race with measure $M(\R^n)Q(\cdot)$ and $U_j \Uniform[0,1]$ i.i.d. for each $(T_j, X_j) \in R$. Consider the following adaptive thinning procedure, subsample all points of $R$ that satisfy $U_j \leq g_i(X_j)/(g(X_j) M(\R^n))$ where $g_i(X_j)$ is defined according to the refinement scheme in $\osstar$, but relative to the history of \emph{points subsampled from $R$ in the order of their acceptance}. It is possible to show that the sequence of accepted points $\{(T_i, X_i, U_i)\}_{i=1}^{\infty}$ have the same marginal distribution as the sequence of proposals in $\osstar$. 

Thus, we can see $\osstar$ and $\rej$ as two separate procedures on the same realization of $R$. For the termination result, notice that $\rej$ considers at least as many points as $\osstar$.  For partial correctness, the points $(T_i, X_i, U_i)$ such that $U_i < f(X_i)/g_i(X_i)$ are exactly the subsampled points that would have resulted from thinning $R$ directly with probability $f(x)/(g(x)M(\R^n))$. Thus, by the Accept-Reject Lemma, the returned values $(T, X)$ will be the first arrival of an exponential race with measure $P$.
\end{proof}

A* sampling ($\astar$ for short) is a branch and bound routine that finds the first arrival of a perturbed exponential race. It follows $\per$ in principle by maintaining a lower bound on all future perturbed arrivals. The difference is that $\astar$ maintains a piecewise constant lower bound over a partition of space that it progressively refines. On every iteration it selects the subset with smallest lower bound, samples the next arrival in that subset, and refines the subset unless it can terminate. It continues refining until the earliest perturbed arrival is less than the minimum of the piecewise constant lower bound. The name A* sampling is a reference to A* search \citep{astar}, which is a path finding algorithm on graphs that uses a best-first criteria for selecting from heuristically valued nodes on the fringe of a set of visited nodes. A* sampling was originally introduced by \cite{maddison2014astarsamp} as an algorithm that maximizes a perturbed Gumbel process. We define it over an exponential race for the sake of consistency. Usually, it is better to work with a Gumbel process to avoid numerical issues.

In more detail, $\astar$ searches over a simulation of an exponential race organized into a space partitioning tree, as in the right hand plot of Figure \ref{fig:eralt}, for the first arrival of the perturbed race. The tree is determined by the splitting function $\mathrm{split}(B, x)$. Each node $v$ of the tree is associated with a subset $B_v \subseteq \R^n$ and an arrival $(T_v, X_v)$ from an exponential race with measure $Q$. 
$\astar$ iteratively expands a subtree of internal visited nodes, taking and visiting one node from the current fringe at each iteration. The fringe $\Lc$ of the visited subtree is always a partition of $\R^n$. Each subset $B \in \Lc$ is associated with the arrival time $T$ of the next arrival of the race with measure $Q$ in $B$. Therefore $T/M(B)$ is a lower bound on all future perturbed arrivals in $B$. $\Lc$ is implemented with a priority queue that prioritizes the subset $B$ with the lowest regional bound $T/M(B)$. As $\astar$ expands the set of visited nodes the lower bound $\min \Lc$ increases.

\begin{algorithm}[t]
\label{alg:astar}
\caption{A* sampling for $P$ with proposal $Q$} \label{alg:astar}
\begin{algorithmic}
	\State $\Lc, \Uc = \mathrm{minPriorityQueue}(), \mathrm{minPriorityQueue}()$ 
	\State $T_1 \sim \Exponential(Q(\R^n))$
	\State $\Lc.\mathrm{pushWithPriority}(T_1/M(\R^n), \R^n)$
	\For{$i=1$ to $\infty$}
		\State $(T_i/M(B_i), B_i) = \Lc.\mathrm{pop}()$
		\State $X_i \sim Q(\cdot \cap B)/Q(B_i)$
		\State $\Uc.\mathrm{pushWithPriority}(T_ig(X_i)/f(X_i), X_i)$
		\State $E \sim \Exponential(Q(B_i))$
		\State $T = T_i + E$
		\If{$\min(\min \Lc, T/M(B_i)) < \min \Uc$}
			\State $\Cc = \mathrm{split}(B_i, X_i)$
			\While{$\Cc \neq \emptyset$}
				\State $C \sim \proba(C) \propto Q(C)$ for $C \in \Cc$
	    			\State $\Lc.\mathrm{pushWithPriority}(T/M(C), C)$
				\State $\Cc = \Cc - \{C\}$
    				\State $E \sim \Exponential(\sum_{C \in \Cc} Q(C))$
    				\State $T = T + E$
			\EndWhile
		\Else
			\State $\Lc.\mathrm{pushWithPriority}(T/M(B_i), B_i)$
		\EndIf
		\If{$\min \Lc \geq \min \Uc$}
			\State\Return $\Uc.\mathrm{pop}()$
		\EndIf
	\EndFor
\end{algorithmic}
\end{algorithm}

$\Lc$ is initialized with the root of the tree $\{(T_1/M(\R^n), \R^n)\}$. At the start of an iteration $\astar$ removes and visits the subset $(T_i/M(B_i), B_i)$ with lowest lower bound on $\Lc$. Visiting a subset begins by realizing a location $X_i$ from $Q(\cdot \cap B_i)/Q(B_i)$ and pushing the perturbed arrival $(T_i g(X_i)/f(X_i), X_i)$ onto another priority queue $\Uc$. $\Uc$ prioritizes earlier arrivals by the perturbed arrival times $T_ig(X_i)/f(X_i)$. In this way $\astar$ decreases the upper bound $\min \Uc$ at each iteration.

$\astar$ attempts to terminate by simulating the next arrival time $T > T_i$ in $B_i$ of the race with measure $Q$. If $\min \Uc \leq \min(\min \Lc, T/M(B_i))$, then the top of $\Uc$ will not be superseded by future perturbed arrivals and it will be the first arrival of the perturbed race. If termination fails, $\astar$ refines the the partition by splitting $B_i$ into a partition $\mathrm{split}(B_i, X_i)$ of children. Arrival times for each of the children are assigned respecting the constraints of the exponential race in $B_i$. Each child $C$ is pushed onto $\Lc$ prioritized by its lower bound $T/M(C)$. Because the lower bounds have increased there is a second opportunity to terminate before continuing. $\astar$ checks if $\min \Uc \leq \min \Lc$, and otherwise continues, see \Algo{alg:astar}. As with $\per$, $\astar$ can be turned into a generator for iterating in order through the points of the perturbed race by replacing the {\bf return} statement with a {\bf yield} statement in \Algo{alg:astar}.

\begin{theorem}[Correctness of A* sampling]
\label{thm:astar}
Let $K(\astar)$ be the number of proposal samples considered before termination. Then
\begin{align*}
\proba(K(\astar) > k) \leq (1-\rho)^k \text{ where } \rho = \frac{P(\R^n)}{Q(\R^n)M(\R^n)}
\end{align*}
and upon termination the return values $(T, X)$ of A* sampling are independent and
\begin{align*}
T \sim \Exponential(P(\R^n)) \quad X \sim \frac{P(\cdot)}{P(\R^n)}.
\end{align*}
\end{theorem}
\begin{proof}
Adapted from \cite{maddison2014astarsamp}. The proposals are generated lazily in a space partitioning tree. If $\{(T_i, X_i)\}_{i=1}^{\infty}$ are the arrivals at every node of the infinite tree sorted by increasing $T_i$, then $(T_i, X_i)$ forms an exponential race with measure $Q$.

For the termination result, each node $v$ of the tree can be associated with a subset $B_v$ and a lower bound $T_v/M(B_v)$. One of the nodes will contain the first arrival of the perturbed process with arrival time $T^*$. $\astar$ visits at least every node $v$ with $T_v/M(B_v) > T^*$. If $M(B)$ is replaced with a constant $M(\R^n)$, then this can only increase the number of visited nodes. The last step is to realize that $\astar$ searching over a tree with constant bounds $M(\R^n)$ searches in order of increasing $T_v$, and so corresponds to a realization of $\per$. The distribution of runtimes of $\per$ is given in Lemma \ref{lem:constantbound}.

For partial correctness, let $(T, X)$ be the return values with highest priority on the upper bound priority queue $\Uc$. The arrival time of unrealized perturbed arrivals is bounded by the lower bound priority queue $\Lc$. At termination $T$ is less than the top of the lower bound priority queue. So no unrealized points will arrive before $(T, X)$. By Lemma \ref{lem:perturb} $(T, X)$ is the first arrival of an exponential race with measure $P$.
\end{proof}

\begin{table}[t]
\begin{center}
\begin{tabular}{llllrr}
\toprule
$P$ & $Q$ & $\Omega$ & $N$ & $\bar{K}(\osstar)$ & $\bar{K}(\astar)$  \\
\midrule
clutter posterior & prior & $\R$ & 6 & 9.34 &7.56 \\
clutter posterior & prior & $\R^2$ & 6 & 38.3 & 33.0  \\
clutter posterior & prior & $\R^3$ & 6 & 130 & 115 \\
robust Bayesian regression & prior & $\R$ & 10  & 9.36 & 6.77\\
robust Bayesian regression & prior & $\R$ & 100 & 40.6 & 32.2 \\
robust Bayesian regression & prior & $\R$ & 1000 & 180 & 152 \\
fully connected Ising model & uniform & $\{-1,1\}^{5}$ &  - & 4.37 & 3.50 \\
fully connected Ising model & uniform & $\{-1,1\}^{10}$ & - & 19.8 & 15.8 \\
\midrule
\end{tabular}
\end{center}
\caption{Comparing $\astar$ and $\osstar$. Clutter and robust Bayesian regression are adapted from \cite{maddison2014astarsamp} and the Ising model from \cite{kim2016lprelaxsamp}. $\Omega$ is the support of the distribution; $N$ is the number of data points; and $\bar{K}(\osstar)$ and $\bar{K}(\astar)$ are averaged over 1000 runs. More information in the Appendix.\label{tab:expts}}.
\end{table}

\subsection{Runtime of A* sampling and OS*}

$\astar$ and $\osstar$ are structurally similar; both search over a partition of space and refine it to increase the probability of terminating. They will give practical benefits over rejection sampling if the bounds $M(B)$ shrink as the volume of $B$ shrinks. In this case the bound on the probability of rejecting $k$ proposals given in Theorems \ref{thm:osstar} and \ref{thm:astar} can be very loose, and $\osstar$ and $\astar$ can be orders of magnitude more efficient than rejection sampling. Still, these methods scale poorly with dimension.

The cost of running $\astar$ and $\osstar$ will be dominated by computing the ratio of densities $f(x)/g(x)$ and computing bounds $M(B)$. Because the number of bound computations is within a factor of 2 of the number of density computations, the number of evaluations of $f(x)/g(x)$ (equivalently number of proposals) is a good estimate of complexity. \tab{tab:expts} presents a summary of experimental evidence that $\astar$ makes more efficient use of density computations across three different problems. For each problem the full descriptions of $P$, $Q$, $M(B)$, and $\mathrm{split}(B, x)$ are found in the Appendix.

The dominance of $\astar$ in experiments is significant, because it has access to {the same information} as $\osstar$. There are at least two factors that may give $\astar$ this advantage. First, if all lower bounds increase sharply after some exploration $\astar$ can retroactively take advantage of that information, as in Section \ref{subsec:adaptivebounds}. Second, $\astar$ can take advantage of refined bound information on the priority queue $\Lc$ before proposing the next sample.
Still, the difference in search strategy and termination condition may counteract these advantages, so a rigorous theory is needed to confirm exactly the sense in which $\astar$ and $\osstar$ differ.
We refer readers to \cite{maddison2014astarsamp} for more detailed experiments.

\section{Conclusion}
\label{sec:con}

The study of Poisson processes is traditionally motivated by their application to natural phenomenon, and Monte Carlo methods are developed specifically for them \citep{ripley1977modelling, geyer1994simulation}. We considered the inverse relationship, using Poisson processes to better understand Monte Carlo methods. We suspect that this general perspective holds value for future directions in research.

Monte Carlo methods that rely on bounds are not suitable for most high dimensional distributions. Rejection sampling scales poorly with dimensionality. Even for A* sampling there are simple
examples where adaptive bounds become uninformative in high dimensions, such as sampling from the uniform hypersphere when using
hyperrectangular search subsets. Still, specialized algorithms for limited classes of distributions may be able to take advantage of conditional
independence structure to improve their scalability.

Another direction is to abandon the idea of representing arbitrary distributions, and study the class of distributions represented by the maxima of combinations of lower order Gumbel processes. This is the approach of the perturbation models studied in Papandreou and Yuille; Gane et al.; Hazan and Jaakkola; Tarlow et al.; and Keshet at al. of this book. In these models a Gumbel process over a discrete space is replaced by sums of independent Gumbel processes over discrete subspaces. The maxima of these models form a natural class of distributions complete with their own measures of uncertainty. An open direction of inquiry is developing efficient algorithms for optimizing their continuous counterparts.

Our study of Poisson processes and Monte Carlo methods was dominated by the theme of independence; the points of an exponential race arrive as independent random variables and accept-reject or perturb do not introduce correlations between the points of the transformed race. Continuing in this direction it is natural to investigate whether other Poisson process models or other operations on an exponential race could be used to define a new class of Monte Carlo methods. In a separate direction the Markov Chain Monte Carlo (MCMC) methods produce a sequence of correlated samples whose limiting distribution is the distribution of interest. The theory of point processes includes a variety of limit theorems, which describe the limiting distribution of random countable sets \citep{daley2007introduction}. It would be interesting to see whether a point process treatment of MCMC bears fruit, either in unifying our proof techniques or inspiring new algorithms.

\section*{Acknowledgements}
We would like to thank Daniel Tarlow and Tom Minka for the ideas, discussions, and support throughout this project. Thanks to the other editors Tamir Hazan and George Papandreou. Thanks to Jacob Steinhardt, Yee Whye Teh, Arnaud Doucet, Christian Robert for comments on the draft. Thanks to Sir J.F.C. Kingman for encouragement. This work was supported by the Natural Sciences and Engineering Research Council of Canada.

\section*{Appendix}
\subsection*{Proof of Lemma \ref{lem:bernoulli}}
\begin{proof}
The lemma is trivial satisfied for $k=0$. For $k > 0$ and $B_i \subseteq B$ we will express
\begin{align}
\label{eq:condoncount}
\proba(\{X_i \in B_i\}_{i=1}^k | N(B) = k)
\end{align}
in terms of counts. The difficulty lies in the possible overlap of $B_i$s, so we consider $2^k$ sets of the form
\begin{align*}
A_j = B_1^* \cap B_2^* \cap \ldots \cap B_k^*
\end{align*}
where $*$ is blank or a complement, and $A_1$ is interpreted as $B \cap B_1^c \cap \ldots \cap B_k^c$. The $A_j$ are a disjoint partition of $B$,
\begin{align*}
B_i = \cup_{j \in I(i)} A_j, \quad B = \cup_{j=1}^{2^k} A_j,
\end{align*}
where $I(i) \subseteq \{1, \ldots, 2^k\}$ is some subset of indices. Let $\Ic = I(1) \times I(2) \times \ldots \times I(k)$, so that each $s \in \Ic$ is a vector indices $(s_1, s_2, \ldots, s_k)$ associated with the disjoint events $\{X_i \in A_{s_i}\}_{i=1}^k$. Thus,
\begin{align*}
\proba(\{X_i \in B_i\}_{i=1}^k | N(B) = k) = \sum_{s \in \Ic} \proba(\{X_i \in A_{s_i}\}_{i=1}^k | N(B) = k).
\end{align*}
For $s \in \Ic$, let $n_j(s) =  \# \{i : s_i = j\}$ be the number of indices in $s$ equal to $j$ and notice that $\sum_{j=1}^{2^k} n_j(s) = k$. To relate the probability if  specific numbering $\{X_i \in A_{s_i}\}_{i=1}^k$ with counts $\{N(A_j) = n_j(s)\}_{j=1}^{2^k}$, we discount by all ways of the arranging $k$ points that result in the same counts.
\begin{align*}
\proba(\{X_i \in A_{s_i}\}_{i=1}^k | N(B) = k) &= \frac{\prod_{j=1}^{2^k}n_j(s)!}{k!}\frac{\proba(\{N(A_j) = n_j(s)\}_{j=1}^{2^k})}{\proba(N(B) = k)}\\
&= \frac{\prod_{j=1}^{2^k} \mu(A_j)^{n_j(s)} }{\mu(B)^k}.
\end{align*}
Thus (\ref{eq:condoncount}) is equal to
\begin{align*}
\sum_{s \in \Ic} \frac{\prod_{j=1}^{2^k} \mu(A_j)^{n_j(s)}}{\mu(B)^k} &= \prod_{i=1}^k \frac{\sum_{j \in I(i)} \mu(A_j)}{\mu(B)} = \prod_{i=1}^k \frac{\mu(B_i)}{\mu(B)}
\end{align*}
\end{proof}
\subsection*{Clutter posterior}
This example is taken exactly from \cite{maddison2014astarsamp}. The clutter problem \citep{minka2001expectation} is to estimate the mean $\theta \in \R^n$ of a Normal distribution under the assumption that some points are outliers. The task is to sample from the posterior $P$ over $w$ of some empirical sample $\{(x_i)\}_{i=1}^N$. 
\begin{align*}
f_i(\theta) &= \frac{0.5 \exp(- 0.5 \lVert \theta - x_i\rVert^2 ) }{(2\pi)^{n/2}} +  \frac{0.5\exp(- 0.5 \lVert x_i\rVert^2 /100^2) }{100^n (2\pi)^{n/2}}\\
\log g(\theta) &=  -\frac{\lVert\theta\rVert^2}{8} \quad \log f(\theta) = \log g(\theta) + \sum_{i=1}^N \log f_i(\theta)\\
(a, b] &= \{y : a_d < y_d \leq b_d\} \text{ for } a, b \in \R^n\\
M((a, b]) &= \prod_{i=1}^N f_i(x^*(a, b, x_i)) \quad x^*(a, b, x)_d = \begin{cases}
a_d & \text{if } x_d < a_d\\
b_d & \text{if } x_d > b_d\\
x_d & \text{o.w. }\\
\end{cases}\\
\mathrm{split}((a,b], x) &= \{(a,b] \cap \{y : y_s \leq x_s\}, (a,b] \cap \{y : y_s > x_s\}\}\\
&\text{where } s = \argmax_d b_d - a_d
\end{align*}
Our dataset was 6 points $x_i \in \R^n$ of the form $x_i = (a_i, a_i, \ldots, a_i)$ for $a_i \in \{-5, -4, -3 , 3, 4, 5\}$. 
\subsection*{Robust Bayesian regression}
This example is an adaption from \cite{maddison2014astarsamp} with looser bounds. The model is a robust linear regression $y_i = w x_i + \epsilon_i$ where the noise $\epsilon_i$ is distributed as a standard Cauchy and $w$ is a standard Normal. The task is to sample from the posterior $P$ over $w$ of some empirical sample $\{(x_i, y_i)\}_{i=1}^N$. 
\begin{align*}
\log g(w) &=  -\frac{w^2}{8}\\
\log f(w) &= \log g(w)  - \sum_{i=1}^N \log(1 + (wx_i - y_i)^2)\\
M((a, b]) &= \prod_{i=1}^N M_i((a,b]) \quad M_i((a, b]) = \begin{cases}
\exp(a) & \text{if } y_i/x_i < a\\
\exp(b) & \text{if } y_i/x_i > b\\
\exp(y_i/x_i) & \text{o.w. }\\
\end{cases}\\
\mathrm{split}((a, b], x) &= \{(a, x], (x, b]\}
\end{align*}
The dataset was generated by setting $w^* = 2$; $x_i \sim \Normal(0, 1)$ and $y_i = wx_i + \epsilon$ with $\epsilon \sim \Normal(0, 0.1^2)$ for $i \leq N/2$; and $x_i = x_{i - N/2}$ and $y_i = - y_{i-N/2}$ for $i > N/2$. 
\subsection*{Attractive fully connected Ising model}
This is an adaptation of \cite{kim2016lprelaxsamp}. The attractive fully connected Ising model is a distribution over $x \in \{-1, 1\}^n$ described by parameters $w_{ij} \sim \Uniform[0, 0.2]$ and $f_i \sim \Uniform[-1, 1]$.
\begin{align*}
\log g(x) &=  0\\
\log f(x) &= \sum_i f_i x_i + \sum_{i < j \leq n} w_{ij} x_ix_j
\end{align*} 
We considered subsets of the form $B = \{x : x_i = b_i, i \in I\}$ where $I \subseteq \{1, \ldots, n\}$ and $b_i \in \{0, 1\}$. We split on one of the unspecified variables $x_i$ by taking variable whose linear program relaxation was closest to 0.5.
\begin{align*}
\mathrm{split}(B, x) = \{B \cap \{x : x_i = 0\}, B \cap \{x : x_i = 1\}\}
\end{align*}
$\log M(B)$ is computed by solving a linear program relaxation for the following type of integer program. Let $b_i \in \{0, 1\}$ for $1 \leq i \leq n$ and $b_{ijkl} \in \{0, 1\}$ for $1 \leq i < j \leq n$ and $k, l \in \{0, 1\}$.
\begin{align*}
\min_x \sum_i -f_i b_i + f_i (1-b_i)  + \sum_{1\leq i < j \leq n} \sum_{k,l \in \{0,1\}} (-1)^{kl + (1-l)(1-k)}w_{ij}b_{ijkl}
\end{align*} 
subject to the constraints for $1 \leq i < j \leq n$, 
\begin{align*}
\sum_{l \in \{0,1\}} b_{ij0l} = 1 - b_i \quad \sum_{k \in \{0,1\}} b_{ijk0} = 1 - b_j\\
\sum_{l \in \{0,1\}} b_{ij1l} = b_i \quad \sum_{k \in \{0,1\}} b_{ijk1} = b_j
\end{align*}
as the subsets $B$ narrowed we just solved new linear programs with constants for the fixed variables.

\bibliography{maddison_chapter}

\end{document}